\documentclass[10pt,reqno,oneside]{amsart}

\usepackage{cancel,bm, amssymb}
\usepackage{color}

\usepackage{graphicx,framed,wrapfig,caption}
\usepackage[colorlinks=true, pdfstartview=FitV, linkcolor=blue, citecolor=blue, urlcolor=blue]{hyperref}
\definecolor{shadecolor}{rgb}{0.95, 0.95, 0.86}

\def \gg{{\mathfrak g}}
\def \wt{\widetilde}

\renewcommand{\d}{{\mathrm d}}
\newcommand{\im}{\mathrm{i}}

\numberwithin{equation}{section}

\newtheorem{theo}{Theorem}[section]
\newtheorem{example}[theo]{Example}
\newtheorem{exercise}[theo]{Exercise}

\newtheorem{lemma}[theo]{Lemma}
\newtheorem{remark}[theo]{Remark}
\newtheorem{problem}[theo]{Riemann-Hilbert Problem}

\newtheorem{proposition}[theo]{Proposition} 
\newtheorem{corollary}[theo]{Corollary} 
\newtheorem{definition}[theo]{Definition}

\def\le{\left}
\def\ri{\right}

\def\un{\underline}
\def\res{\mathop{\mathrm {res}}\limits}

\def\bt{\begin{theo}}
\def\et{\end{theo}}
\def\bc{\begin{corollary}}
\def\ec{\end{corollary}}
\def\bx{\begin{example}\small}
\def\ex{\end{example}}
\def\bxr{\begin{exercise}\small}
\def\exr{\end{exercise}}
\def\bl{\begin{lemma}}
\def\el{\end{lemma}}
\def\bd{\begin{definition}}
\def\ed{\end{definition}}
\def\bp{\begin{proposition}}
\def\ep{\end{proposition}}

\def\br{\begin{remark}}
\def\er{\end{remark}}

\def\be{\begin{equation}}
\def\ee{\end{equation}}
\def\ov {\overline}

\def\&{\hspace{-15pt}&}
\def\mod{\, \mathrm{mod}\,\,}
\def\bea{\begin{eqnarray}}
\def\eea{\end{eqnarray}}
\def\beas{\begin{eqnarray*}}
\def\eeas{\end{eqnarray*}}

\def \pa{\partial}
\def\C{{\mathbb C}}
\def\L{\mathcal L}
\def\R{{\mathbb R}}

\def\wh{\widehat}

\def\Z{{\mathbb Z}}

\def\l{\lambda}

\def\1{{\bf 1}}

\def\Roots{\mathfrak R}

\def\z{\zeta}

\def\eqref#1{(\ref{#1})}
\def\bt{{\mathbf t}}

\def\binom#1#2{\le({#1 \atop #2}\ri)}

\begin{document}

\title[Generalized Vorob'ev-Yablonski Polynomials]{Hankel Determinant Approach to Generalized Vorob'ev-Yablonski Polynomials and their Roots}

\author{Ferenc Balogh}
\address{Centre de recherches math\'ematiques,
Universit\'e de Montr\'eal, C.~P.~6128, succ. centre ville, Montr\'eal,
Qu\'ebec, Canada H3C 3J7 and,
Department of Mathematics and
Statistics, Concordia University, 1455 de Maisonneuve W., Montr\'eal, Qu\'ebec,
Canada H3G 1M8}
\email{Ferenc.Balogh@concordia.ca}

\author{Marco Bertola}
\address{Centre de recherches math\'ematiques,
Universit\'e de Montr\'eal, C.~P.~6128, succ. centre ville, Montr\'eal,
Qu\'ebec, Canada H3C 3J7 \and
Department of Mathematics \and
Statistics, Concordia University, 1455 de Maisonneuve W., Montr\'eal, Qu\'ebec,
Canada H3G 1M8 \and, 
Sector of Mathematical Physics, 
SISSA/ISAS 
via Bonomea, 265  }
\email{Marco.Bertola@concordia.ca, mbertola@sissa.it}

\author{Thomas Bothner}
\address{Centre de recherches math\'ematiques,
Universit\'e de Montr\'eal, C.~P.~6128, succ. centre ville, Montr\'eal,
Qu\'ebec, Canada H3C 3J7 and,
Department of Mathematics and
Statistics, Concordia University, 1455 de Maisonneuve W., Montr\'eal, Qu\'ebec,
Canada H3G 1M8}
\email{bothner@crm.umontreal.ca}



\date{\today}

\begin{abstract}
Generalized Vorob'ev-Yablonski polynomials have been introduced by Clarkson and Mansfield in their study of rational solutions of the second Painlev\'e hierarchy. We present new Hankel determinant identities for the squares of these special polynomials  in terms of Schur polynomials. As an application of the identities, we analyze the roots of generalized Vorob'ev-Yablonski polynomials and provide formul\ae\, for the boundary curves of the highly regular patterns observed numerically in \cite{CM}.
\end{abstract}
\maketitle

\section{Introduction and statement of results}
Let $u=u(x;\alpha)$ denote a solution of the second Painlev\'e equation
\begin{equation}\label{i:1}
	u_{xx}=xu+2u^3+\alpha,\ \ \ x\in\mathbb{C}.
\end{equation}
It is known that for special values of the parameter $\alpha\in\mathbb{C}$ the equation admits  rational solutions.  In fact Vorob'ev and Yablonski \cite{V,Ya} showed that for $\alpha=n\in\mathbb{Z}$, the equation \eqref{i:1} has a unique rational solution of the form
\begin{equation}\label{i:2}
	u(x;n) = \frac{\d}{\d x}\ln\left\{\frac{\mathcal{Q}_{n-1}(x)}{\mathcal{Q}_n(x)}\right\},\ \ n\in\mathbb{Z}_{\geq 1};\ \ \ \ \ \ \ u(x;0)=0,\ \ \ \ \ \ \ u(x;-n)=-u(x;n),\ \ n\in\mathbb{Z}_{\geq 1},
\end{equation}
which is constructed in terms of the Vorob'ev-Yablonski polynomials $\{\mathcal{Q}_n(x)\}_{n\geq 0}$. These special polynomials can be defined via a differential-difference equation
\begin{equation}\label{i:0}
	\mathcal{Q}_{n+1}(x)\mathcal{Q}_{n-1}(x)=x\mathcal{Q}_n^2(x)-4\left(\mathcal{Q}_n''(x)\mathcal{Q}_n(x)-\big(\mathcal{Q}_n'(x)\big)^2\right),\ \ n\in\mathbb{Z}_{\geq 1},\ \ x\in\mathbb{C},
\end{equation}
where $\mathcal{Q}_0(x)=1,\mathcal{Q}_1(x)=x$, or equivalently \cite{KO} in determinantal form: with $q_k(x)=0$ for $k<0$,
\begin{equation}\label{i:3}
	\mathcal{Q}_n(x) = \prod_{k=1}^n\frac{(2k)!}{2^kk!}  \det\Big[q_{n-2\ell+j}(x)\Big]_{\ell,j=0}^{n-1},\ \ n\in\mathbb{Z}_{\geq 1};\ \ \ \ \ \sum_{k=0}^{\infty}q_k(x)w^k=\exp\left[-\frac{4}{3}w^3+wx\right].
\end{equation}
For our purposes, it will prove useful to rewrite \eqref{i:3} in terms of Schur polynomials. In general (cf. \cite{M}), the Schur polynomial $s_{\lambda}\in\mathbb{C}[\mathbf{t}]$ in the variable $\mathbf{t}=(t_1,t_2,t_3,\ldots), t_j\in\mathbb{C}$ associated to the partition $\lambda=(\lambda_1,\lambda_2,\ldots,\lambda_{\ell(\lambda)})$ with $\mathbb{Z}\ni\lambda_j\geq\lambda_{j+1}>0$ is determined by the Jacobi-Trudi determinant,
\be
\label{JTrudi}
s_{\lambda}(\bt) = \det\big[h_{\lambda_j-j+k}(\bt)\big]_{j,k=1}^{\ell(\lambda)}\ .
\ee
Here, $h_k(\mathbf{t})$ for $k\in\mathbb{Z}_{\geq 0}$ is defined by the generating series
\be
\label{hdef}
\sum_{k=0}^{\infty}h_k(\bt)z^{k}=\exp\left(\sum_{j=1}^{\infty}t_j z^j\right); \hspace{0.5cm} \ \textnormal{and}\ \ \ \ h_k(\bt )=0,\ \ k<0.
\ee
\br
[Homogeneity]
From \eqref{hdef} it follows immediately that $h_k(\bt)$ is a weighted-homogeneous function, 
\be\nonumber
h_k(\bt) = 
\epsilon^k h_k \le(\epsilon^{-1} t_1, \epsilon^{-2} t_2, \epsilon^{-3} t_3,\dots \ri),\ \ \ \epsilon \in \C\backslash\{0\},
\ee
and hence also 
\be
\label{homogSchur}
s_\l (\bt) = \epsilon^{|\l|} s_\l \le(\epsilon^{-1} t_1, \epsilon^{-2} t_2, \epsilon^{-3} t_3,\ldots \ri),\ \ \ \ |\lambda|=\sum_{j=1}^{\ell(\lambda)}\lambda_j.
\ee
\er
For the special choice of a staircase partition, 
\begin{equation}\label{stair}
	\lambda\equiv\delta_n=(n,n-1,n-2,\ldots,2,1);\ \ \ \ \ell(\delta_n)=n,
\end{equation}
the identities \eqref{JTrudi},\eqref{hdef} and \eqref{i:3} lead to the representation of $Q_n(x)$ in terms of Schur polynomials,
\begin{equation*}
	\mathcal{Q}_n(x)=\prod_{k=1}^n\frac{(2k)!}{2^kk!}\,s_{\delta_n}\left(x,0,-\frac{4}{3},0,0,\ldots\right),\ \ x\in\mathbb{C},\ \ n\in\mathbb{Z}_{\geq 1}.
\end{equation*}
It is well known that equation \eqref{i:1} admits higher order generalizations and itself forms the first member of a full hierarchy. To be more precise, let $\mathcal{L}_N$ denote the following quantities expressed in terms of the  Lenard recursion operator,
\begin{equation}
\label{Lenardrec}
	\frac{\d}{\d x} \mathcal{L}_{N+1}[u]=\left(\frac{\d^3}{\d x^3}+4u\frac{\d}{\d x}+2u_x\right)\mathcal{L}_N[u],\ \ N\in\mathbb{Z}_{\geq 0};\ \ \ \mathcal{L}_0[u]=\frac{1}{2} ,
\end{equation}
and with the integration constant determined uniquely by the requirement $\L_n[0]=0,\ n\geq 1$.
The recursion gives, for instance, 
\begin{equation*}
	\mathcal{L}_1[u]=u,\ \ \ \ \mathcal{L}_2[u]=u_{xx}+3u^2,\ \ \ \ \mathcal{L}_3[u]=u_{xxxx}+5(u_x)^2+10uu_{xx}+10u^3.
\end{equation*}
The $N$-th member of the Painlev\'e II hierarchy is subsequently defined as the ordinary differential equation
\begin{equation}\label{i:5}
	\left(\frac{\d}{\d x}+2u\right)\mathcal{L}_N\big[u_x-u^2\big]=xu+\alpha_N,\ \ \ x\in\mathbb{C},\ \ \alpha_N\in\mathbb{C};\ \ \ \ u=u(x;\alpha_N,N).
\end{equation}
Hence, the first member $N=1$ is Painlev\'e II \eqref{i:1} itself, and more generally, the $N$-th member is an ordinary differential equation  of order  $2N$. 
Besides \eqref{i:5}, we shall also consider a case which involves additional complex parameters $t_3,t_5,\ldots,t_{2N-1}$. With $u=u(x;\alpha_N,\un t,N)$ for $x,\alpha_N\in\mathbb{C}$ and $\un t=(t_3,\ldots,t_{2N-1})\in\mathbb{C}^{N-1}$,
\be
\label{genPIIhier}
\left(\frac{\d}{\d x}+2u\right)\mathcal{L}_N\big[u_x-u^2\big]= \sum_{k=1}^{N-1} (2k+1) t_{2k+1}\le( \frac{\d}{\d x} +2u\ri) \L_k\big[u_x-u^2\big] + xu+ \alpha_N.
\ee
For \eqref{i:5} and \eqref{genPIIhier}, it is known \cite{G,DK2} that rational solutions exist if and only if $\alpha_N=n\in\mathbb{Z}$. Moreover, Clarkson and Mansfield in \cite{CM} introduced generalizations of the Vorob'ev-Yablonski polynomials for $N=2,3$ which allow to compute the rational solutions of \eqref{i:5} once more in terms of logarithmic derivatives,
\begin{equation*}
	u(x;n,N)=\frac{\d}{\d x}\ln\left\{\frac{\mathcal{Q}_{n-1}^{(N)}(x)}{\mathcal{Q}_n^{[N]}(x)}\right\},\ n\in\mathbb{Z}_{\geq 1};\hspace{0.5cm}u(x;0,N)=0,\ \ \ u(x;-n,N)=-u(x;n,N),\ \ n\in\mathbb{Z}_{\geq 1}.
\end{equation*}
This approach has been extended to \eqref{genPIIhier} for general $N\in\mathbb{Z}_{\geq 1}$ by Demina and Kudryashov \cite{DK1, DK2} who found in particular the analogues of \eqref{i:0} for, what we shall call {\it generalized Vorob'ev-Yablonski polynomials} $\mathcal{Q}_n^{[N]}(x;\un t)$,
\begin{eqnarray}
	\mathcal{Q}_{n+1}^{[N]}(x;\un t)\mathcal{Q}_{n-1}^{[N]}(x;\un t)&=&\big(\mathcal{Q}_n^{[N]}(x;\un t)\big)^2\Bigg\{x-2\mathcal{L}_N\left[2\frac{\d^2}{\d x^2}\ln\mathcal{Q}_n^{[N]}(x;\un t)\right]\label{diffrel}\\
	&&\hspace{0.5cm}+2\sum_{k=1}^{N-1}(2k+1)t_{2k+1}\mathcal{L}_k\left[2\frac{\d^2}{\d x^2}\ln\mathcal{Q}_n^{[N]}(x;\un t)\right]\Bigg\},\ \ n\in\mathbb{Z}_{\geq 1}\nonumber
\end{eqnarray}
with $\mathcal{Q}_0^{[N]}(x;\un t)=1$ and $\mathcal{Q}_1^{[N]}(x;\un t)=x$. For fixed $\un t=(t_3,t_5,\ldots,t_{2N-1})\in\mathbb{C}^{N-1}$ and $n,N\in\mathbb{Z}_{\geq 1}$ these special polynomials are then used in the construction of the unique rational solutions of \eqref{genPIIhier},
\begin{equation*}
	u(x;n,\un t,N)=\frac{\d}{\d x}\ln\left\{\frac{\mathcal{Q}_{n-1}^{[N]}(x;\un t)}{\mathcal{Q}_n^{[N]}(x;\un t)}\right\};\hspace{0.5cm}u(x;0,\un t,N)=0,\ \ \ u(x;-n,\un t,N)=-u(x;n,\un t,N).
\end{equation*}
\subsection{Determinantal identities}
It is mentioned in \cite{DK1}, but not proven, that also $\mathcal{Q}_n^{[N]}(x;\un t)$ can be expressed as a Schur polynomial. In our first Theorem below we shall close this small gap.

\begin{figure}
\includegraphics[width=0.329\textwidth]{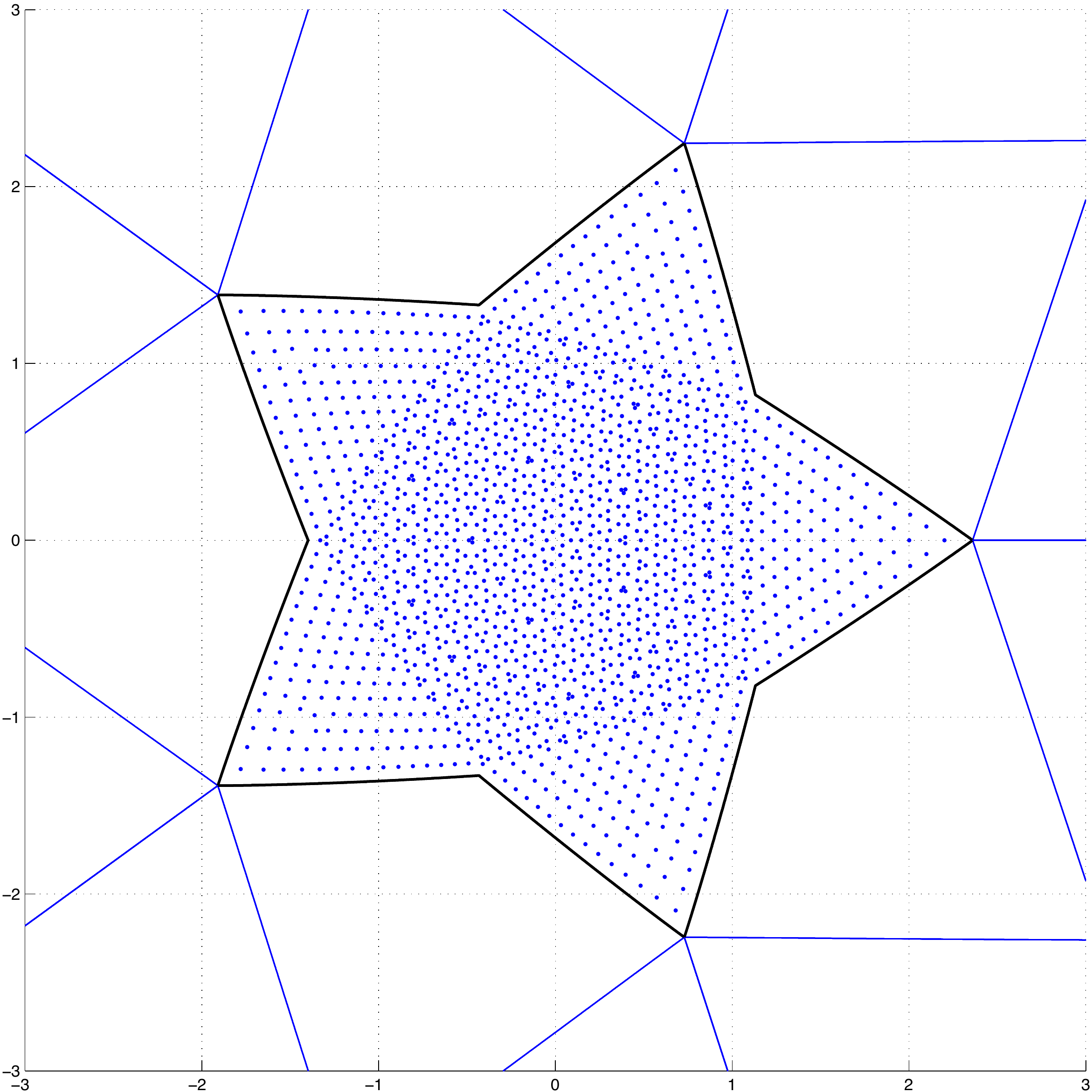}
\includegraphics[width=0.329\textwidth]{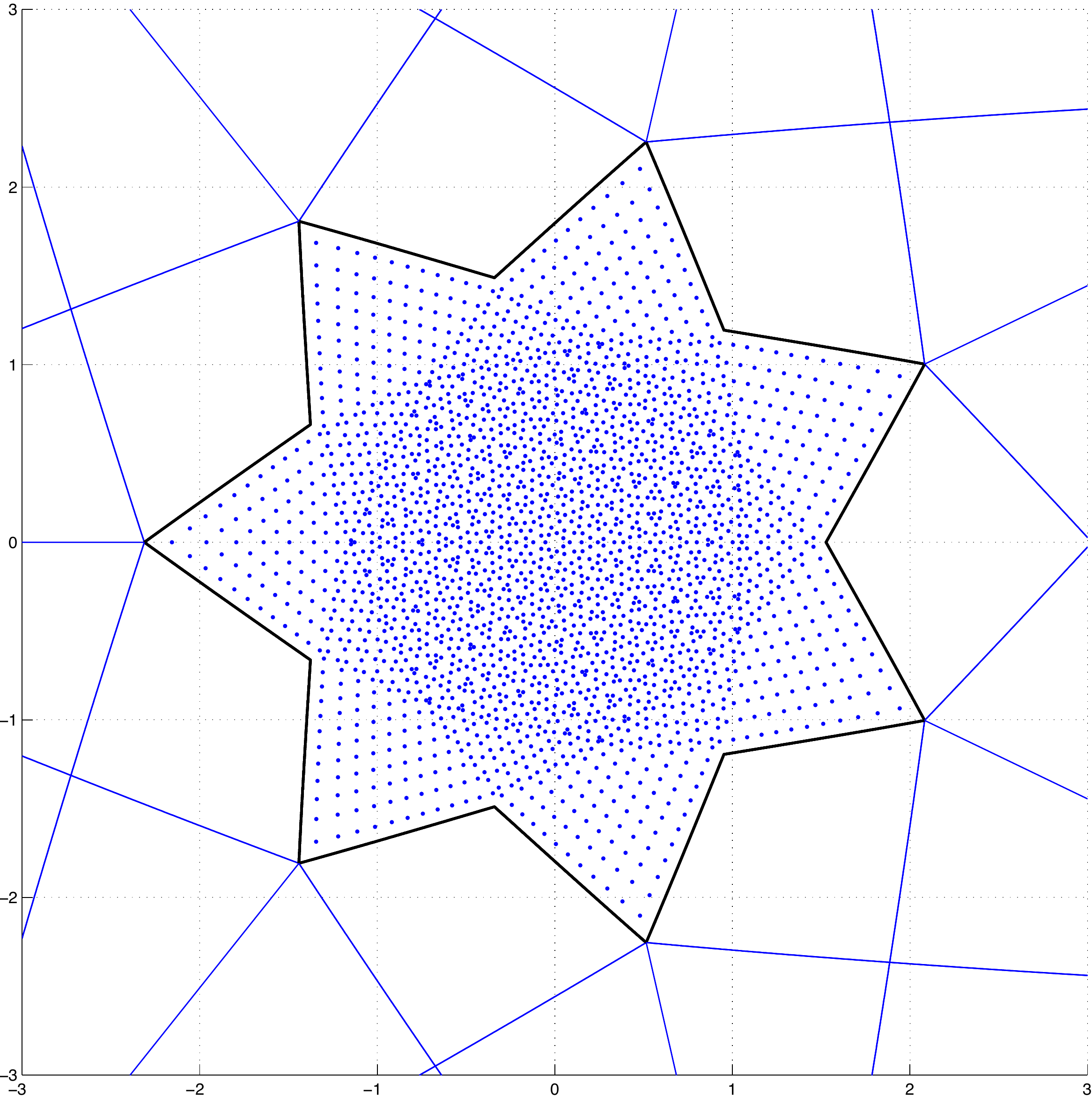}
\includegraphics[width=0.329\textwidth]{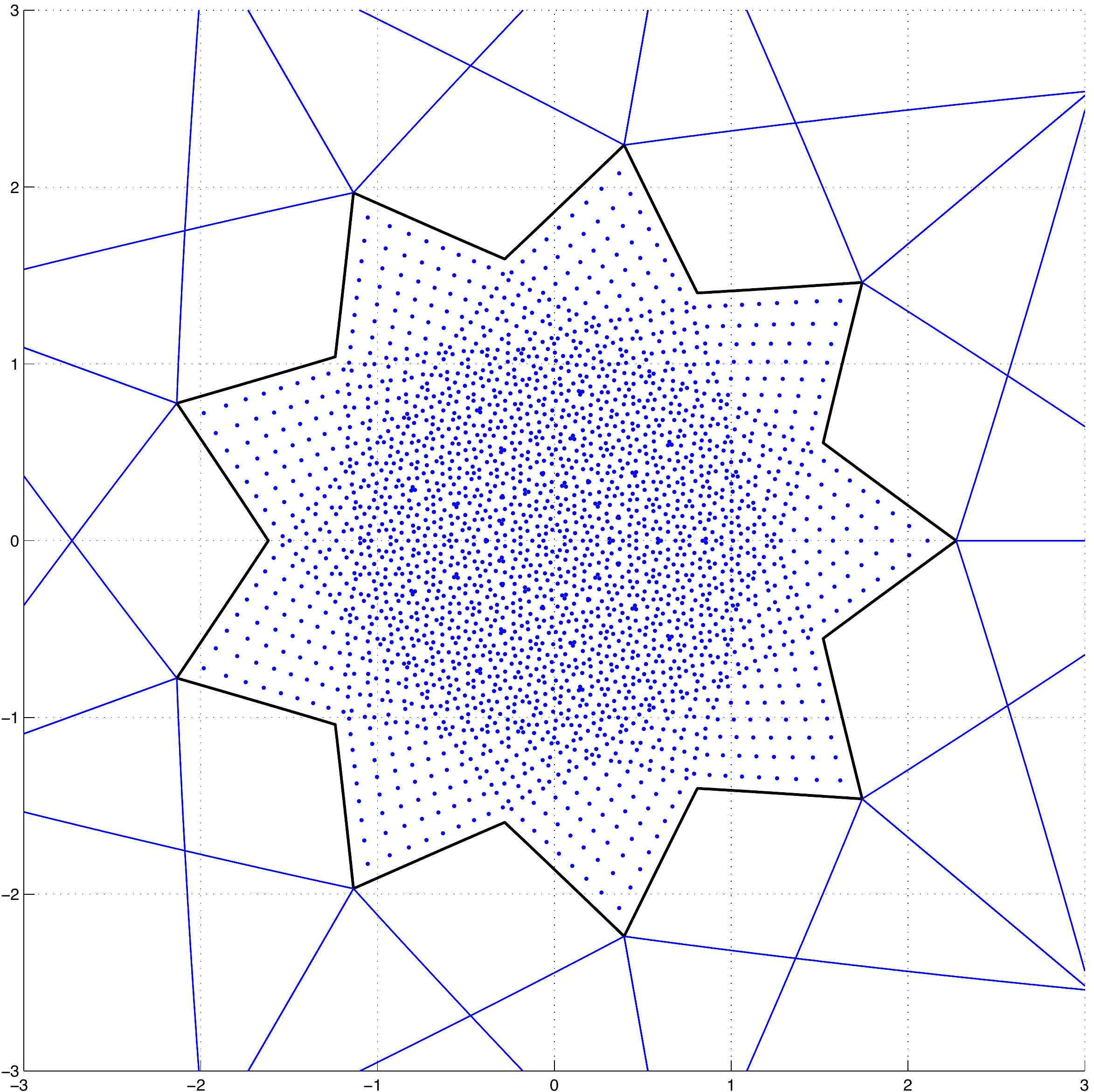}
\includegraphics[width=0.329\textwidth]{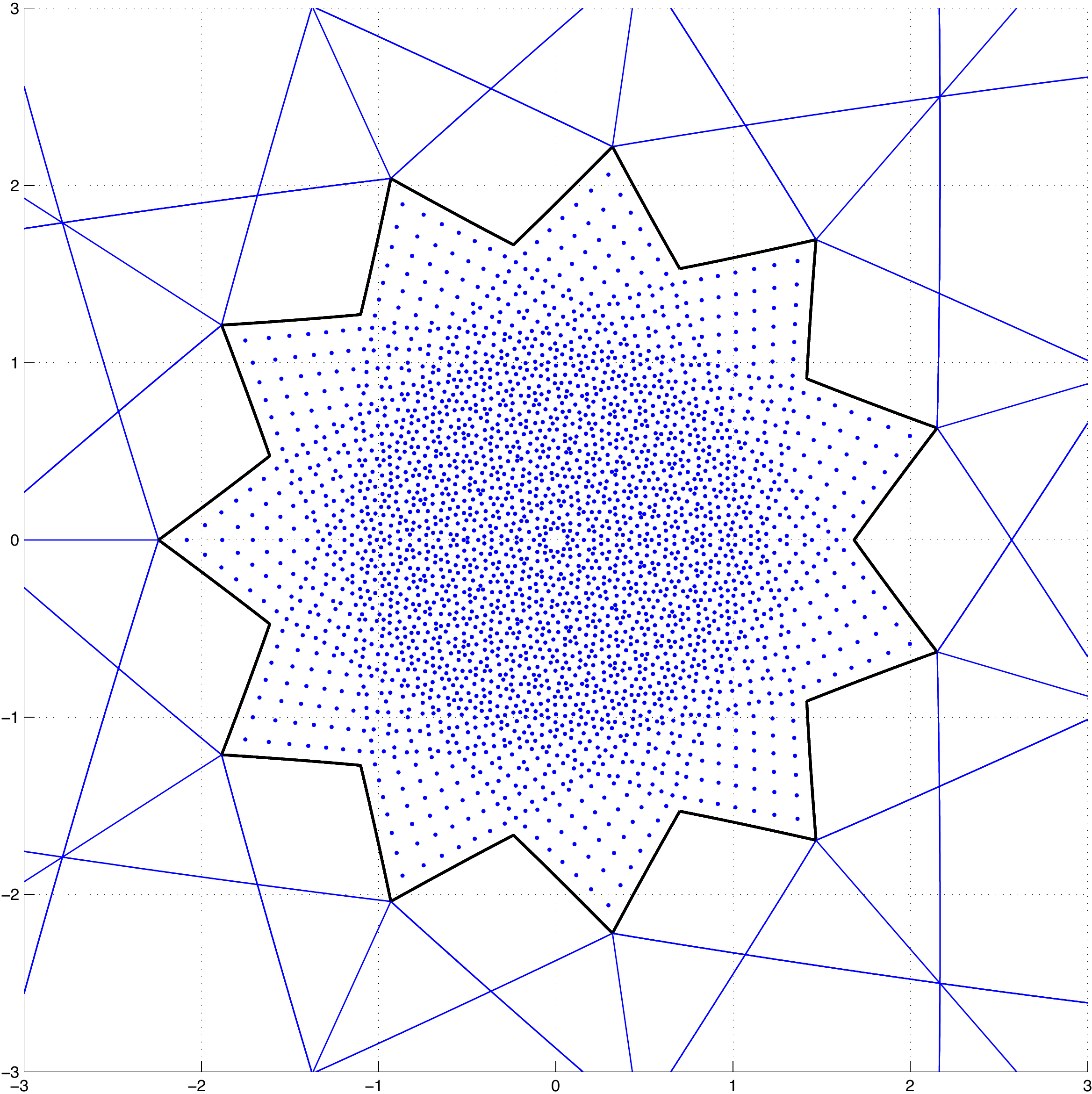}
\includegraphics[width=0.329\textwidth]{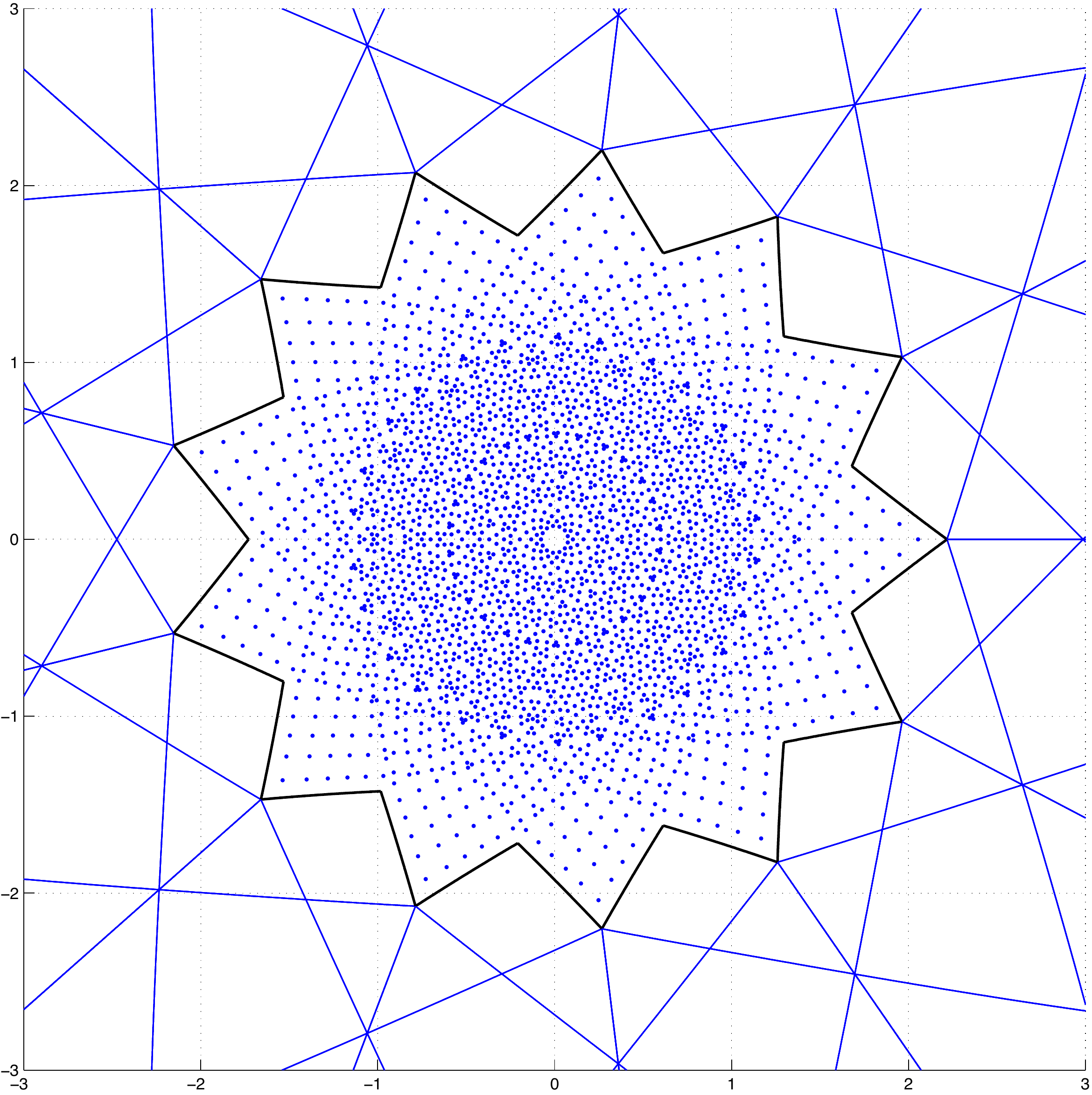}
\includegraphics[width=0.329\textwidth]{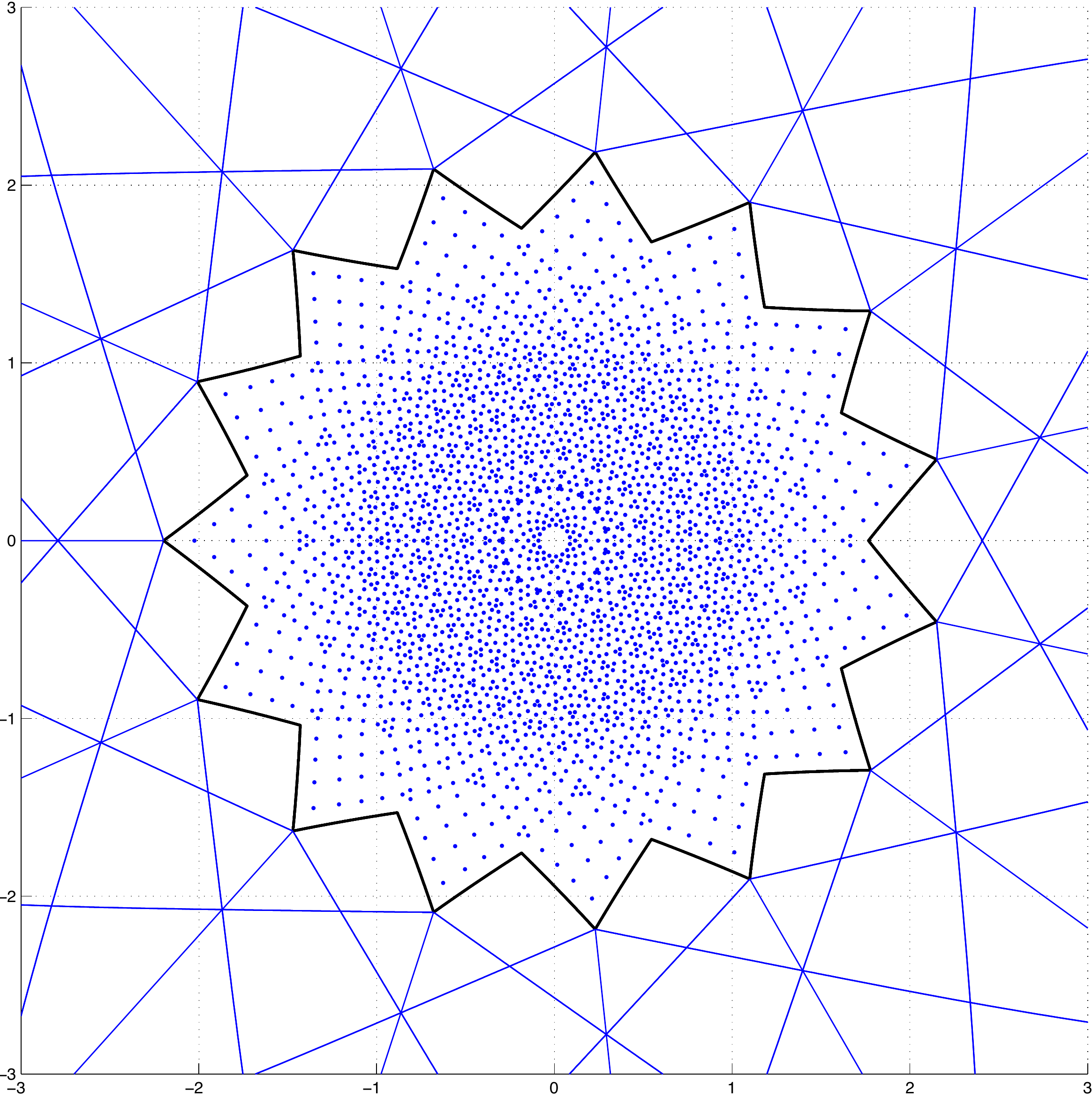}
\caption{The roots of the rescaled higher Vorob'ev-Yablonski polynomials $\mathcal{Q}^{[2]}_{60}$,  $\mathcal{Q}^{[3]}_{70}$,  $\mathcal{Q}^{[4]}_{72}$, $\mathcal{Q}^{[5]}_{77}$, $\mathcal{Q}^{[6]}_{78}$, $\mathcal{Q}^{[7]}_{75}$ (from left to right and top to bottom). See \eqref{higher}, \eqref{roots} for their definition.  
  The symmetry of the pattern is easily explained from the definition of the polynomials. The locations of the outer vertices of the star shaped regions are given in \eqref{starcorners}. The various lines that appear in the Figures are not straight lines but real analytic arcs defined by the implicit equation \eqref{condzero}. It is quite evident that for $N\geq 2$ there are further subdivisions of the star-shaped region into subregions.
  }
\label{Stars}
\end{figure}

\begin{theo}\label{small} Let $\delta_n$ denote the staircase partition \eqref{stair} of length $n\in\mathbb{Z}_{\geq 1}$. For any 
\begin{equation*}
	\un t=(t_3,t_5,\ldots,t_{2N-1})\in\mathbb{C}^{N-1},
\end{equation*}
the generalized Vorob'ev-Yablonski polynomial $\mathcal{Q}_n^{[N]}(x;\un t),x\in\mathbb{C}$ defined in \eqref{diffrel} equals
\begin{equation}\label{Schurid}
	\mathcal{Q}_n^{[N]}(x;\un t)=\prod_{k=1}^n\frac{(2k)!}{2^kk!}\,s_{\delta_n}\left(x,0,2^2t_3,0,2^4t_5,\ldots,2^{2N}t_{2N+1},0,0,0,\ldots\right),\ \ \ \ t_{2N+1}\equiv -\frac{1}{2N+1}.
\end{equation}
\end{theo}
Besides the Jacobi-Trudi type identity \eqref{i:3}, Vorob'ev-Yablonski polynomials can also be expressed as Hankel determinants, in fact in \cite{BB} the following Hankel determinant representation for the squared polynomial $\mathcal{Q}_n^2(x)$ was obtained,
\be\label{the1}
  \mathcal{Q}_{n}^2(x) = (-1)^{\binom{n+1}{2}}\frac{1}{2^{n}}\prod_{k=1}^{n}\left[\frac{(2k)!}{k!}\right]^2\,\det\big[\mu_{\ell+j-2}(x)\big]_{\ell,j=1}^{n+1},\ \ \ x\in\mathbb{C}
\ee
with $\{\mu_k(x)\}_{k\in\mathbb{Z}_{\geq 0}}$ defined by the generating function
\begin{equation*}
	\exp\left[xw -\frac {w^3}3\right] =\sum_{j=0}^\infty \mu_j(x) w^j.
\end{equation*}
In our second Theorem we present the analogue of \eqref{the1} for the generalized Vorob'ev-Yablonski polynomial $\mathcal{Q}_n^{[N]}(x;\un t)$.
\begin{theo}\label{Fesq} Let $\un t=(t_3,\ldots,t_{2N-1})\in\mathbb{C}^{N-1}$ and $n\in\mathbb{Z}_{\geq 1}$. For any $x\in\mathbb{C}$ we have the Hankel determinant representation
\begin{equation}\label{Fesqid}
	\Big(\mathcal{Q}_n^{[N]}(x;\un t)\Big)^2=(-1)^{\binom{n+1}{2}}\frac{1}{2^{n}}\prod_{k=1}^{n}\left[\frac{(2k)!}{k!}\right]^2\,\det\Big[\mu_{\ell+j-2}^{[N]}(\bt_o)\Big]_{\ell,j=1}^{n+1}
\end{equation}
where we use the abbreviation
\begin{equation*}
	\bt_o=\big(t_1,0,t_3,0,t_5,\ldots,t_{2N-1},0,t_{2N+1},0,0,0,\ldots\big);\ \ \ \ \ t_1=x,\ \ \ \ t_{2N+1}=-\frac{1}{2N+1}
\end{equation*}
and the coefficients $\{\mu_j^{[N]}(\bt_o)\}_{j\in\mathbb{Z}_{\geq 0}}$ are defined by the generating function
\begin{equation}\label{muk}
	\exp\left[\sum_{j=1}^{\infty}t_jw^j\right]=\sum_{k=0}^{\infty}\mu_k^{[N]}(\bt_o)w^k,\ \ \ \ \ \ t_j\equiv 0,\ j>2N+1.
\end{equation}
\end{theo}
\br
In fact, the statement of Theorem \ref{Fesq} is the specialization of a more general identity for Schur functions (compare Lemma \ref{lemmadouble} below) which in our case reads 
\be
	s_{\delta_n}^2\big(t_1,0,{2^2}t_3,0,{2^4}t_5,0,\ldots\big)= 
 {2^{n^2}}s_{(n+1)^{n}}\left(t_1,0,t_3,0,t_5,0,t_7,\ldots\right). 
\ee 
Here, $\lambda=(n+1)^n$ denotes the rectangular partition with $n+1$ rows of length $n$ and the specialization consists in simply setting
\begin{equation*}
	t_1=x,\ \ \ \ \ t_{2N+1}=-\frac{1}{2N+1},\ \ \ \ \ t_j\equiv 0,\ j>2N+1.
\end{equation*}
\er
\subsection{Roots of higher Vorob'ev-Yablonski polynomials}
In analogy to \cite{BB}, we provide a direct application of Theorem \ref{Fesq}. Numerical studies carried out in \cite{KNFH,CM,DK2} show that the zeros of generalized Vorob'ev-Yablonski polynomials form highly regular and symmetric patterns as can be clearly seen in Figure \ref{Stars}. These patterns in case of the Painlev\'e II equation itself have been first analyzed in \cite{BuM1,BuM2}. However, the approach outlined in \cite{BB} starts directly from \eqref{the1} and not from a Lax pair associated with \eqref{i:1}. To be more precise, the identity \eqref{Fesqid} allows us to localize the roots of  the generalized  Vorob'ev-Yablonski polynomials as $n\to\infty$  by analyzing associated pseudo-orthogonal polynomials. Of course in the generalized case these patterns depend on the parameters $\{t_{2j+1}\}_{j=1}^{N-1}$ (compare \cite{DK2}); we shall confine ourselves here to the case of {\em higher Vorob'ev-Yablonski polynomials}, namely the case 
\begin{equation}
\label{higher}
	t_3=t_5=\ldots=t_{2N-1}=0.
\end{equation}
More specifically, we are considering the roots of the rescaled higher Vorob'ev-Yablonski polynomials 
\be\label{roots}
\Roots_n^{[N]} = \le\{
x\in\mathbb{C}:\,\,\mathcal{Q}_n^{[N]}\le(n^{\frac {2N}{2N+1}} x \ri) =0
\ri \}.
\ee
These sets admit a discrete $\Z_{2N+1}$ rotational symmetry, which follows immediately from Theorem \ref{small} and the homogeneity \eqref{homogSchur};
\be
\mathcal{Q}_{n}^{[N]}( \omega x)= \omega^{\frac {n}{2}(n+1)} \mathcal{Q}_n^{[N]}(x),\hspace{1cm}\omega= {\rm e}^{\frac {2\pi\im}{2N+1}}.
\ee
We can provide a partial  analytic description for the boundary of the polygons $P_N$ seen in Figure \ref{Stars} which asymptotically contain the sets \eqref{roots} as $n\to\infty$. More precisely we have first the following Theorem.

\begin{theo}
There exists a compact region $P_N$ in the complex $x$--plane such that for any $\delta>0$ the root sets $\Roots_n^{[N]}$ are contained in a $\delta$-neighborhood $\mathcal N_\delta(P_N)$ of $P_N$ as $n\rightarrow\infty$.
\end{theo}

The description of the regions $P_N$ is provided in part  by Theorem \ref{thm:boundary} below. First we require 
\bd Given $N\in\mathbb{Z}_{\geq 1}$, let  $a=a(x;N),x\in\mathbb{C}$ denote the unique solution of the algebraic equation
\begin{equation}\label{i:8}
	(2a)^{2N+1}-x(2a)^{2N}+(-1)^N\binom{2N}{N}=0
\end{equation}
which is analytic in the domain 
\be\nonumber
x\in\C\Big\backslash\bigcup_{k=0}^{2N}\Big[0, x_k^{[N]}\Big]
\ee
and behaves near $x=\infty$ as 
\begin{equation}\label{pick}
	a=\frac{x}{2}+\mathcal{O}\left(x^{-2N}\right),\ \ \ x\rightarrow\infty.
\end{equation}
Here, the points $x=x_k^{[N]},k=0,\ldots,2N$  are the solutions of
\begin{equation}
\label{starcorners}
	x^{2N+1}=(-1)^{N}(2N+1)\left(\frac{2N+1}{2N}\right)^{2N}\binom{2N}{N},
\end{equation}
and form the outer vertices of the regular star-shaped regions shown in Figure \ref{Stars}.
\ed

\begin{theo}
\label{thm:boundary}
 The regions  $P_N$ are compact,  invariant under the rotations of angle $\frac {2\pi}{2N+1}$,  contain the origin and their boundary $\partial P_N$ consists of branches of the locus in the complex $x$-plane described by 
\begin{equation}\label{boundaries}
\mathfrak Z_N= \le\{ x\in\C: \ \ \Re\big(\varphi(z;a)\big)\Big|_{z=z_k^{[N]}}=0\ri\}.
\end{equation}
Here  $z=z_k^{[N]},k=1,\ldots,2N$ are the solutions of the equation
\begin{equation*}
	z^{2N}-\frac{1}{2a}T_{N-1,-\frac{1}{2}}\left(\frac{z^2}{a^2}\right)=0, 
\end{equation*}
where $T_{m,\alpha}(z)$ denotes the Maclaurin polynomial of degree $m\in\mathbb{Z}_{\geq 0}$ of the function $(1+z)^{\alpha}=1+\mathcal{O}(z),z\rightarrow 0$.
Moreover  $a=a(x)$ is defined in \eqref{i:8} and \eqref{pick}, and the function $\varphi$ is defined by
\begin{equation*}
	\varphi(z;a)=-2\ln\left(\frac{z+(z^2+a^2)^{\frac{1}{2}}}{\im a}\right)+\frac{2}{z}\big(z^2+a^2\big)^{\frac{1}{2}}-\frac{1}{2N+1}\frac{(z^2+a^2)^{\frac{3}{2}}}{a^3z^{2N+1}} T_{N-1,-\frac{3}{2}}\left(\frac{z^2}{a^2}\right)
\end{equation*}
with principal branches for fractional exponents and logarithms. 
\end{theo}
The branches of the real--analytic curves specified by $\mathfrak Z_N$ of Theorem \ref{thm:boundary} are plotted as the arcs in Figure \ref{Stars}. Perhaps more important than what Theorem \ref{thm:boundary} above says, is what it does not say. In fact of all the branches of curves defined by \eqref{boundaries} we are not able to effectively discern which ones actually form the boundary of $P_N$. In particular we cannot conclude in general that the points \eqref{starcorners} belong to $\pa P_N$. 
\br
A local analysis (which we do not propose here but is essentially identical to \cite{BB}) shows that the angles between consecutive  arcs emanating from the points $x_{k}^{[N]}$ \eqref{starcorners}  is $\frac {2\pi}5$.
\er

\subsection{The roots inside $P_N$}
Inspection of Figure \ref{Stars} clearly shows that the pattern of roots within $P_N$ is subdivided in subregions.
This can be easily {\em qualitatively} understood in terms of the steepest descent analysis; the so--called $\gg$-function of the problem (see Section \ref{character} below) is an Abelian integral on a Riemann surface of genus $0$ on the outside of $P_N$ and of genera $2,4,\dots$ inside. In fact we can show that $x=0$ belongs to a region where the genus is $2N$ and thus it is reasonable to deduce that there are nested regions of higher and higher genus, until the maximum is reached ($2N$). These regions are quite evident in Figure \ref{Stars}. 
In principle the boundaries between these nested regions could be described as well in terms of Abelian integrals, but it is beyond the scope of this paper to attempt any such detailed description. 

\subsection{Outline of paper} We conclude the introduction with a short outline of the upcoming sections. First Theorem \ref{small} is derived in Section \ref{PIIhier} by referring to the KdV and mKdV hierarchies for which we construct a rational tau function in terms of Schur polynomials. Subsequently an explicit scaling reduction brings us back to the Painlev\'e II hierarchy and Theorem \ref{small} follows. After that we turn towards Theorem \ref{Fesq}, but opposed to the proof of \eqref{the1} in \cite{BB} which relied on \eqref{i:0}, identity \eqref{Fesqid} will follow from Schur function identities and Theorem \ref{small}. In the final Section \ref{character} we follow largely the logic outlined in \cite{BB}. However we choose not to present any details on the nonlinear steepest descent analysis for the underlying orthogonal polynomials. Once the correct inequalities for the $\mathfrak{g}$-function have been verified the asymptotic analysis outside of $P_N$ is almost identical to \cite{BB}, see Section \ref{character} for further details.
\subsection{Acknowledgments} 
All authors are grateful to P. Clarkson for useful discussions about this project. M.B. is supported
in part by the Natural Sciences and Engineering Research Council of Canada. F. B. and M. B. are supported in part by the Fonds de recherche Nature et technology du Qu\'ebec. T.B. acknowledges hospitality of SISSA, Trieste in February 2015.
Early stages of the manuscript were carried out while F.B. was a Postdoctoral fellow at SISSA.
\section{Short reminder about the (m)KdV  and Painlev\'e\ II hierarchies}
\label{PIIhier}
The goal of this section is to remind the reader very briefly of the construction of the Painlev\'e\ II hierarchy as a scaling reduction of the modified Korteweg-de Vries (mKdV)  hierarchy, cf. \cite{CM}. In doing so we will en route derive Theorem \ref{small}. 
\subsection{The KdV hierarchy} 
The KdV hierarchy involves the Lenard recursion operator 
\begin{equation}\label{Lenard}
	\frac{\pa}{\pa x} \L_{n+1}[u] = \le(\frac{\pa^3 }{\pa x^3 }+ 4u(x)\frac{\pa}{\pa x} + 2u_x(x)\ri) \L_n[u],\ \ \ \L_0[u]=\frac{1}{2},\ \ \ \L_n[0]=0
\end{equation}
and its equations are written as the partial differential equations
\begin{equation}\label{KdV}
	\frac{\pa u}{\pa t_{2n+1}} = \frac \pa{\pa x} \L_{n+1}[u],\ \ \ n\in\mathbb{Z}_{\geq 0};\ \ \ u=u(\bt_o),\ \ \bt_o=(t_1,0,t_3,0,t_5,\ldots).
\end{equation}
It is customary, and we will adhere to the custom, to denote the variable $t_1$ by $x$ since $\L_1[u] = u$ and hence the first member of the hierarchy above reads simply $\pa_{t_1} u = u_x$. In general, the equations of the hierarchy should be viewed as an infinite set of compatible evolution equations for a single function $u = u(x)$.  A {\em solution} of the hierarchy is then a function $u(x;t_3,t_5,\dots)$.\smallskip

\bd A function $\tau_{_{\textnormal{KdV}}}=\tau_{_{\textnormal{KdV}}}(\bt_o)$ is called a tau function for the KdV hierarchy \eqref{KdV} if the function 
\begin{equation*}
	u(\bt_o) = 2 \frac {\pa^2} {\pa x^2} \ln \tau_{_{\textnormal{KdV}}}(\bt_o)\ ,\ \ \ x \equiv t_1.
\end{equation*}
solves the hierarchy \eqref{KdV}. 
We note that multiplication by an arbitrary constant (in $x$) of $\tau_{_{\textnormal{KdV}}}$ gives another tau function. 
\ed
\subsection{Rational solutions to KdV and staircase Schur polynomials} The solutions of the KdV equation rational in $x$ for all values of $t=t_3$ (and for all higher times $t_5, t_7,\dots$) and vanishing at $x=\infty$ were completely characterized in \cite{AMM}; they all belong to the countable union of orbits flowing out of initial data of the form 
\be
u_n(x,0,0,\dots)= \frac{n(n+1)}{x^2}\ ,\ \ \ n\in \mathbb{Z}_{\geq 0}\ .
\ee
The corresponding tau functions $\tau_n(\bt_o)$ were obtained explicitly in \cite{AM} in terms of Wronskians of certain polynomials in $\bt_0$. Up to normalization and re-parametrization these Wronskians coincide with Schur polynomials associated to staircase partitions evaluated at the odd times, namely
\be
\label{adler_moser}
\tau_n(\bt_o) = s_{\delta_n}(t_1,0,2^2t_3,0,2^4 t_5,\dots)\ ,\ \ \ u_n(\bt_o) = 2 \frac {\pa^2} {\pa x^2} \ln \tau_n(\bt_o)\ ,
\ee
where $\delta_n$ denotes the staircase partition \eqref{stair} of length $n\in\mathbb{Z}_{\geq 1}$.  Moreover, it can be shown (cf.~\cite{SW}) that these are the only Schur polynomials
that give KdV tau functions when all even times are set to zero.
\br
The particular rescaling $t_{2\ell+1}\mapsto 2^{2\ell}t_{2\ell+1}$ is used in \eqref{adler_moser} in order to correct the normalization so that the coefficients in \eqref{KdV} are as indicated.
\er

\subsection{The mKdV hierarchy}
The {\em modified} KdV (mKdV) hierarchy is defined in terms of a new dependent variable $v = v(\bt_o)$ which is related to $u$ via the Miura transformation 
\begin{equation}\label{Miura}
	u = \mp v_x - v^2,
\end{equation}
where the choice of signs is arbitrary. More is true: if $v$ satisfies $u = -v_x -v^2$,  then the new function $\wh u= v_x -v^2$ is a {\em different} solution of the KdV hierarchy (and vice versa); this is an example of a B\"acklund transformation. Inserting \eqref{Miura} into \eqref{KdV} yields a new set of evolution equations 
\begin{align*}
\frac \pa{\pa t_{2n+1}} \le(\mp v_x - v^2\ri) 
&= \frac \pa{\pa x} \L_{n+1}\big[\mp v_x - v^2\big] 
\stackrel{\eqref{Lenard}}{=} \le(\frac{\pa^3 }{\pa x^3 }- 4(\pm v_x+v^2) \frac{\pa}{\pa x} - 2(\pm v_{xx}+2vv_x)\ri) \L_n\big[\mp v_x-v^2\big]\\
&=\le(\frac \pa{\pa x}  \pm  2v\ri)  \frac \pa{\pa x}\le(\frac \pa{\pa x}  \mp 2v\ri) \L_n\big[\mp v_x-v^2\big].
\end{align*}
This can be rewritten as follows
\begin{equation*}
	\le(\mp\frac \pa{\pa x}  -  2v\ri)  \frac{\pa v} {\pa t_{2n+1}} =\le(\frac \pa{\pa x}\pm 2v\ri)  \frac \pa{\pa x}\le(\frac \pa{\pa x}\mp 2v\ri)  \L_n\big[\mp v_x - v^2\big]
\end{equation*}
or equivalently 
\begin{equation}\label{weakmKdV}
\le(\mp\frac \pa{\pa x}  - 2v\ri) \Bigg\{\underbrace{ \frac{\pa v} {\pa t_{2n+1}} 
-\frac \pa{\pa x}\le( \mp \frac \pa{\pa x}   +  2v\ri) \L_n\big[ \mp v_x-v^2\big]}_{\mathfrak Q_{n}^{(\pm)}[v]}
\Bigg\}  =0.
\end{equation}
We now notice that the two expressions 
\begin{equation}\label{fact}
	\le( \mp \frac \pa{\pa x}   +  2v\ri) \L_n\big[ \mp v_x-v^2\big]=\frac{1}{2}
	\le( \mp  \frac \pa{\pa x}   +  2v\ri)   \le[\int \d x \le(\frac \pa{\pa x}\pm 2v\ri)  \frac \pa{\pa x}\le(\frac \pa{\pa x} \mp 2v\ri)\ri]^{n} 
\end{equation}
define {\em the same} differential polynomial in $v$ since the right hand side is clearly invariant under the map $x \mapsto -x$.  Thus we can simply write 
\begin{equation*}
	\mathfrak Q_{n}^{(+)}[v]= \mathfrak Q_{n}^{(-)}[v]=\mathfrak Q_{n}[v],
\end{equation*}	
omitting the reference to the choice of sign. We now want to conclude that the expression $\mathfrak Q_{n}[v]$ vanishes identically;  the two equations in \eqref{mKdVpm} below are simply stating that $F(x)= \mathfrak Q_{n}[v]$ is a  joint solution of the two ordinary differential equations $(\pm \pa_x + 2v)F(x)=0$. Thus $\mathfrak Q_{n}[v]$ should be the in the null-space of both equations $\pm \pa_x + 2v$; as long as $v$ is not identically zero (which is an un-interesting situation), the only function in both null-spaces is the null function and hence $\mathfrak Q_{n}[v]\equiv 0$. Thus we have concluded that if $u$ is a  solution of the KdV hierarchy \eqref{KdV} and $v$ is related to $u$ by \eqref{Miura}, then $v$ must solve the hierarchy of equations indicated below and named  {\em mKdV hierarchy},
\begin{equation}\label{mKdVpm}
 	\frac{\pa v}{\pa t_{2n+1}} = \frac \pa{\pa x}\le( \mp \frac \pa{\pa x}   +  2v\ri) \L_n\big[ \mp v_x-v^2\big],\ \ n\in\mathbb{Z}_{\geq 0};\ \ \ \ v=v(\bt_o).
\end{equation}
The choice of signs is irrelevant, since the right hand side (as noted above) yields the same differential polynomial in $v$.

\subsection{Schur functions and Painlev\'e II hierarchy} Let us now return to our special situation for which we fix
\begin{equation*}
	 t_1=x,\ \ \ \ t_{2N+1}=-\frac{1}{2N+1},\ \ \ \ \un t=(t_3,t_5,\ldots,t_{2N-1})\in\mathbb{C}^{N-1},\ \ \ \ t_{2j+1}=0,\ \ j>N.
\end{equation*}	
\bp\label{propMiura} For $n,N\in\mathbb{Z}_{\geq 1}$ define the two functions
\begin{equation}\label{gW}
	g_n(x;\un t)=\ln s_{\delta_n}\big(x,0,2^2t_3,0,2^4t_5,\ldots,2^{2N}t_{2N+1},0,0,0,\ldots),\ \ \ \ \ W_n(x;\un t)=g_{n+1}(x;\un t)-g_n(x;\un t)
\end{equation}
with some fixed branch for the logarithm. We then have the Miura relation
\begin{equation}\label{Miura2}
	2\,\partial_x^2\,g_n(x;\un t)=-\partial_x^2 W_n(x;\un t)-\big(\partial_x W_n(x;\un t)\big)^2.
\end{equation}
\ep
A proof of \eqref{Miura2} can be found in Appendix \ref{appmiura}. In view of Proposition \ref{propMiura} we note that the two functions 
\begin{equation}\label{uvdef}
	u(\bt_o)=2\,\partial_x^2\,g_n(x;\un t),\ \ \ \ \ v(\bt_o)=\partial_xW_n(x;\un t)
\end{equation}
satisfy precisely 
the Miura relation \eqref{Miura} with the choice of the minus sign, namely $u  = - v '- v^2$. Since $s_{\delta_n}$ gives a tau function for the KdV hierarchy it follows that $v$ satisfies the hierarchy \eqref{mKdVpm} for $n=0,\ldots,N$. Summarizing
\bp The function
\begin{equation*}
	w(x;\un t)=-\partial_xW_n(x;\un t)
\end{equation*}
satisfies the mKdV hierarchy in the form
\begin{equation}\label{mKdV}
	\frac{\partial w}{\partial t_{2n+1}}=\frac{\partial}{\partial x}\left(\frac{\partial}{\partial x}+2w\right)\mathcal{L}_n\big[w_x-w^2\big],\ \ 0\leq n\leq N.
\end{equation}
\ep
Recalling the homogeneity property \eqref{homogSchur} we see that $w(x;\un t)$ obeys a simple scaling invariance which will allow us to reduce the partial differential equations \eqref{mKdV} to an ordinary differential equation; we carry out a {\it scaling reduction}: 
\begin{enumerate}
	\item[(i)] View $w=-v(\bt_o)$ as a function in the variables $t_1=x,\un t=(t_3,t_5,\ldots,t_{2N-1})\in\mathbb{C}^{N-1}$ and $t_{2N+1}$.
	\item[(ii)] By homogeneity \eqref{homogSchur}, it follows that $w=w(t_1,\un t,t_{2N+1})$ is a function of the form
	\begin{equation}\label{red:1}
		w=\big(-(2N+1)t_{2N+1}\big)^{-\frac{1}{2N+1}}V(T_1,T_3,\ldots,T_{2N-1}),
	\end{equation}
	and $V$ depends on the ``new" variables
	\begin{equation*}
		T_{2k+1}=\frac{t_{2k+1}}{(-(2N+1)t_{2N+1})^{\frac{2k+1}{2N+1}}},\ \ \ \ k=0,\ldots,N-1.
	\end{equation*}
	\item[(iii)] Substituting \eqref{red:1} into the left hand side of \eqref{mKdV} with $n=N$, we find
	\begin{equation}\label{red:2}
		(2N+1)t_{2N+1}\frac{\partial w}{\partial t_{2N+1}}=-\big(-(2N+1)t_{2N+1}\big)^{-\frac{1}{2N+1}}\left[V+\sum_{j=0}^{N-1}(2j+1)T_{2j+1}\frac{\partial V}{\partial T_{2j+1}}\right].
	\end{equation}
	\item[(iv)] Next we evaluate \eqref{red:2},\eqref{red:1} at $t_{2N+1}=-\frac{1}{2N+1}$ and compare the result to \eqref{mKdV},
	\begin{equation}\label{red:3}
		\frac{\partial}{\partial x}\left(\frac{\partial}{\partial x}+2V\right)\L_N\big[V_x-V^2\big]=V+\sum_{j=0}^{N-1}(2j+1)t_{2j+1}\frac{\partial V}{\partial t_{2j+1}}.
	\end{equation}
	\item[(v)] Since $t_1=x$ and $V+x\frac{\partial V}{\partial x} = \frac{\partial}{\partial x}(xV)$, \eqref{red:3} can be rewritten with the help of \eqref{mKdV},
	\begin{equation}\label{red:4}
		\frac{\partial}{\partial x}\left\{\left(\frac{\partial }{\partial x}+2V\right)\L_N\big[V_x-V^2\big]-xV-\sum_{j=1}^{N-1}(2j+1)t_{2j+1}\left(\frac{\partial}{\partial x}+2V\right)\L_j\big[V_x-V^2\big]\right\}=0.
	\end{equation}
\end{enumerate}
Equation \eqref{red:4} is an ordinary differential equation for the function $V=w(x;\un t)$ in which $\un t\in\mathbb{C}^{N-1}$ appear as {\it parameters}. Since
\begin{equation*}
	w(x;\un t)=\partial_x\big(g_n(x;\un t)-g_{n+1}(x;\un t)\big)=-\frac{n+1}{x}+\mathcal{O}\left(x^{-2}\right),\ \ x\rightarrow\infty,
\end{equation*}
integration in \eqref{red:4} yields \eqref{genPIIhier} with $\alpha_N=n+1$. Recall \cite{DK1,DK2} that $\alpha_N\in \Z$ is necessary to have a rational solution to \eqref{genPIIhier} and for all 
integer values of $\alpha_N$ there exists a unique rational solution which can be obtained from the trivial solution for $\alpha_N=0$ by B\"acklund transformations. Therefore we have the following

\begin{theo} For $n,N\in\mathbb{Z}_{\geq 1}$ the unique rational solution of the Painlev\'e II hierarchy \eqref{genPIIhier} is
\begin{equation}\label{red:5}
	u(x;n+1,\un t,N)=\frac{\d}{\d x}\ln\frac{s_{\delta_n}}{s_{\delta_{n+1}}}\big(x,0,2^2t_3,0,2^4t_5,\ldots,2^{2N}t_{2N+1},0,0,0,\ldots\big),\ \ \ t_{2N+1}=-\frac{1}{2N+1}
\end{equation}
and we have the identity
\begin{equation*}
	\mathcal{Q}_n^{[N]}(x;\un t)=\prod_{k=1}^n\frac{(2k)!}{2^kk!}s_{\delta_n}\big(x,0,2^2t_3,0,2^4t_5,\ldots,2^{2N}t_{2N+1},0,0,0,\ldots),\ \ x\in\mathbb{C}.
\end{equation*}
\end{theo}
\begin{proof} It is easy to see that the LHS of \eqref{red:5} is a rational solution to \eqref{genPIIhier} by the scaling reduction \eqref{red:1}-\eqref{red:4}. By the uniqueness of the rational solutions of the Painlev\'e II hierarchy we have 
\begin{equation*}
	\mathcal{Q}_n^{[N]}(x;\un t) = c_{n,N}(\un t)s_{\delta_n}\big(x,0,2^2t_3,0,2^4t_5,\ldots,2^{2N}t_{2N+1},0,0,0,\ldots\big),
\end{equation*}
with an $x$-independent factor $c_{n,N}(\un t)$. However, to leading order,
\begin{equation*}
	s_{\delta_n}\big(x,0,2^2t_3,0,2^4t_5,\ldots,2^{2N}t_{2N+1},0,0,0,\ldots\big)\sim s_{\delta_n}(x,0,0,0,\ldots)=\frac{x^{|\delta_n|}}{h(\delta_n)},\ \ x\rightarrow\infty,
\end{equation*}
where $h(\lambda)$ denotes the product of the hook-lengths of $\lambda$ (cf. \cite{M}). Since
\begin{equation*}
	|\delta_n|=\frac{n}{2}(n+1),\ \ \ \ \ \ h(\delta_n)=\prod_{k=1}^n\frac{2^kk!}{(2k)!}
\end{equation*}
and $\mathcal{Q}_n^{[N]}(x;\un t)$ is a monic polynomial of degree $\frac{n}{2}(n+1)$, the claim follows.
\end{proof}
\section{Proof of Theorem \ref{Fesq}}
We will appeal to certain identities satisfied by  symmetric functions which can be found, for instance, in \cite{M}. First let us start with the following lemma.
\begin{lemma}\label{lemmadouble}
The symmetric polynomial identity
\begin{equation*}
	s_{\delta_n}^2(2^0t_1,0,2^2t_3,0,2^4t_5,0,\ldots)=
2^{-n} s_{(n+1)^n} \left(2t_1,0,2^3t_3,0,2^5t_5,\dots\ri) = 
 {2^{n^2}}s_{(n+1)^{n}}\left(t_1,0,t_3,0,t_5,0,t_7,\dots\right)
\end{equation*}
holds, where $\lambda=(n+1)^{n}$ stands for the rectangular partition with $n+1$ rows of length $n$ and $\delta_n$ is the staircase partition \eqref{stair}. 
\end{lemma}
\begin{proof} The Schur polynomial $s_{\delta_n}$ can be written in terms of the projective Schur polynomial $P_{\delta_n}$ labeled by the same partition
\begin{equation}\label{p:1}
	s_{\delta_n}(t_1,0,t_3,0,t_5,0,\ldots)=P_{\delta_n}(t_1,t_3,t_5,\ldots).
\end{equation}
For a proof of this identity see \cite{M}, $\S$ 3.8, example $3$, page $259$, and also \cite{DW}, Lemma V.4. Second, for a strict partition $\lambda$, i.e. $\lambda_1>\lambda_2>\ldots>\lambda_{\ell(\lambda)}$, we have \cite{Y}, Theorem $4$,
\begin{equation}\label{p:2}
	2^{\ell(\lambda)}P_{\lambda}^2\left(\frac{t_1}{2},\frac{t_3}{2},\frac{t_5}{2},\ldots\right)=s_{\bar{\lambda}}(t_1,0,t_3,0,t_5,\ldots),
\end{equation}
with $\bar{\lambda}$ denoting the double of the partition $\lambda$ which is defined via its Frobenius characteristics,
\begin{equation*}
	\bar{\lambda}=\big(\lambda_1,\lambda_2,\ldots,\lambda_{\ell(\lambda)}\big|\lambda_1-1,\lambda_2-1,\ldots,\lambda_{\ell(\lambda)}-1\big).
\end{equation*}
Combining \eqref{p:1} and \eqref{p:2},
\begin{eqnarray*}
	s_{\delta_n}^2(2^0t_1,0,2^2t_3,0,2^4t_5,0,\ldots) &=&P_{\delta_n}^2(2^0t_1,2^2t_3,2^4t_5,\ldots)=2^{-n}s_{(n+1)^n}(2t_1,0,2^3t_3,0,2^5t_5,\ldots)\\
	&=&2^{n^2}s_{(n+1)^n}(t_1,0,t_3,0,t_5,\ldots)
\end{eqnarray*}
where we have used homogeneity \eqref{homogSchur} in the last step. This concludes the proof.
\end{proof} 
We are now ready to derive Theorem \ref{Fesq} by referring to \eqref{Schurid} and Lemma \ref{lemmadouble}.
\begin{proof}[Proof of Theorem \ref{Fesq}] Let $\un t=(t_3,t_5,\ldots,t_{2N-1})\in\mathbb{C}^{N-1}$ and
\begin{equation*}
	\bt_o=(t_1,0,t_3,0,t_5,\ldots,t_{2N-1},0,t_{2N+1},0,0,0,\ldots),\ \ \ \ t_1=x,\ \ t_{2N+1}=-\frac{1}{2N+1}.
\end{equation*}
This gives us
\begin{eqnarray*}
	\big(\mathcal{Q}_n^{[N]}(x;\un t)\big)^2&\stackrel{\eqref{Schurid}}{=}&\prod_{k=1}^n\left[\frac{(2k)!}{2^kk!}\right]^2\,s_{\delta_n}^2\big(2^0t_1,0,2^2t_3,0,2^4t_5,\ldots,2^{2N}t_{2N+1},0,0,0,\ldots)\\
	&=&\frac{1}{2^n}\prod_{k=1}^n\left[\frac{(2k)!}{k!}\right]^2s_{(n+1)^n}(t_1,0,t_3,0,t_5,\ldots,t_{2N+1},0,0,0,\ldots)\\
	&=&\frac{1}{2^n}\prod_{k=1}^n\left[\frac{(2k)!}{k!}\right]^2s_{(n)^{n+1}}(\bt_o)=\frac{1}{2^n}\prod_{k=1}^n\left[\frac{(2k)!}{k!}\right]^2\det\big[\mu_{n-\ell+j}^{[N]}(\bt_o)\big]_{\ell,j=1}^{n+1}\\
	&=&(-1)^{\binom{n+1}{2}}\frac{1}{2^n}\prod_{k=1}^n\left[\frac{(2k)!}{k!}\right]^2\det\big[\mu_{\ell+j-2}^{[N]}(\bt_o)\big]_{\ell,j=1}^{n+1},
\end{eqnarray*}
where we used that for the transposed partition $\lambda'$,
\begin{equation}\label{transpose}
	s_{\lambda'}(t_1,0,t_3,0,t_5,\ldots)=(-1)^{|\lambda|}s_{\lambda}(-t_1,0,-t_3,-t_5,\ldots)\stackrel{\eqref{homogSchur}}{=}s_{\lambda}(t_1,0,t_3,0,t_5,\ldots),
\end{equation}
and that the Schur polynomials of rectangular partitions are Hankel determinants.
\end{proof}

\bc\label{corfreak} Let $\bt_{o}=(t_1,0,t_3,0,t_5,\ldots)$ and $\{h_k(\bt_{o})\}_{k\in\mathbb{Z}_{\geq 0}}$ as in \eqref{hdef}. Introducing the notation
\begin{equation*}
	\Delta_{n,\ell}(\bt_{o}) = \det\big[h_{j+k-2+\ell}(\bt_{o})\big]_{j,k=1}^{n+1},\ \ n,\ell\in\mathbb{Z}_{\geq 0},
\end{equation*}
we have the Hankel determinant identity
\begin{equation}\label{freak}
	\Delta_{n+1,0}(\bt_{o})=(-1)^n\Delta_{n,2}(\bt_{o}).
\end{equation}
\ec
\begin{proof} Note that
\begin{eqnarray*}
	s_{(n+1)^n}(\bt_{o})&=&\det\big[h_{n+1-j+k}(\bt_{o})\big]_{j,k=1}^n=(-1)^{n-1}\det\big[h_{j+k}(\bt_{o})\big]_{j,k=1}^n,\\
	s_{(n)^{n+1}}(\bt_{o})&=&\det\big[h_{n-j+k}(\bt_{o})\big]_{j,k=1}^{n+1}=(-1)^n\det\big[\mu_{j+k-2}(\bt_{o})\big]_{j,k=1}^{n+1},
\end{eqnarray*}
and since $|(n)^{n+1}|=|(n+1)^n|=n(n+1)\equiv 0\mod 2$, the stated identity follows from \eqref{transpose}.
\end{proof}
\br
Identity \eqref{freak} in Corollary \ref{corfreak} {\em does not hold} if any of the even-index times is nonzero. 
\er

\section{Characterization of the set $\Roots_n^{[N]}$}\label{character}
The logic we are following here is identical to \cite{BB}. The square of the polynomials $\mathcal{Q}_n^{[N]}(x)$ is proportional to a Hankel determinant 
\begin{equation*}
	\Delta_{n}(x;N) = \det\big[\mu_{j+k-2}^{[N]}(\bt_o)\big]_{j,k=1}^{n+1}
\end{equation*}
of the moments $\mu_k^{[N]}(\bt_o)$ \eqref{muk}, which can alternatively be written as 
\be
	\mu_k^{[N]}(\bt_o) = \frac{1}{2\pi\im}\oint_{S}z^k{\rm e}^{\frac{x}{z}-\frac{z^{-2N-1}}{2N+1}} \frac{\d z}{z};\ \ \ \bt_o=(x,0,0,\ldots,0,t_{2N+1},0,0,0,\ldots),\ \ t_{2N+1}=-\frac{1}{2N+1}
\ee
where $S\subset\mathbb{C}$ denotes the unit circle traversed in counterclockwise direction. It is then a well-known fact that $\Delta_{n}(x;N)=0$ if and only if the Riemann--Hilbert problem \ref{master} has no solution, or equivalently, if and only if the $n$-th monic orthogonal polynomial for the weight 
\begin{equation*}
	\d\mu_0 (z;x,N)= \frac { {\rm e}^{-\vartheta(z;x,N)} }{2\pi\im z},\ \ \ \ \ \vartheta(z;x,N) = \frac 1{(2N+1)z^{2N+1}}-\frac{x}{z},\ \ \ z\in S
\end{equation*}
does not exist.
In view of the scaling $x\mapsto n^{\frac {2N}{2N+1}}x$ in \eqref{roots} we also perform a scaling $z\mapsto n^{-\frac 1{2N+1}}z$ so that we arrive at the following Riemann--Hilbert problem with a varying exponential weight.
\begin{problem}\label{master} Suppose $\gamma\subset\mathbb{C}$ is a smooth Jordan curve which encircles the origin counterclockwise. Let $\Gamma=\Gamma(z;x,n,N)       $ denote the $2\times 2$ matrix-valued piecewise analytic function which is uniquely characterized by the following three properties.
\begin{enumerate}
	\item $\Gamma(z)$ is analytic for $z\in\mathbb{C}\backslash\gamma$
	\item Given the orientation of $\gamma$, the limiting values $\Gamma_{\pm}(z)$ from the $(+)$ and $(-)$ side of the contour exist and are related via the jump condition
	\begin{equation*}
		\Gamma_+(z)=\Gamma_-(z)\begin{bmatrix}
		1 & w(z;x,N)\\
		0 & 1
		\end{bmatrix},\ \ z\in\gamma;\ \ \ \ w(z;x,N)=\frac{{\rm e}^{-n \vartheta(z;x,N)}}{2\pi\im z}.
	\end{equation*}
	\item The function $\Gamma(z)$ is normalized as $z\rightarrow\infty$,
	\begin{equation*}
		\Gamma(z)=\left(I+\frac{\Gamma_1(x;n,N)}{z}+\mathcal{O}\left(z^{-2}\right)\right)z^{n\sigma_3},\ \ \ \sigma_3=\begin{bmatrix}
		1 & 0\\
		0 & -1
		\end{bmatrix}.
	\end{equation*}
\end{enumerate}
\end{problem}
Then we have, compare \cite{BB},
\bp
The zeros of the scaled Vorob'ev-Yablonski polynomials $Q_n^{[N]}(n^{\frac {2N}{2N+1}} x)$ coincide with the values of $x$ for which the problem \ref{master} is {\em not} solvable.
\ep
In principle an asymptotic analysis of the Problem \ref{master} as $n\to\infty$ is possible using the Deift-Zhou steepest descent analysis \cite{DZ1,DVZ,DKMVZ}, and the zeros will be located asymptotically in terms of appropriate Theta functions as in \cite{BB}. However here we simply want to prove the {\em absence} of zeros outside of a certain compact region $P_N$ and, {\em en route}, give a partial  characterization of the boundary $\pa P_N$.  For a more comprehensive analysis which is only marginally different from the present situation we refer to \cite{BB}; here we shall just remind the reader that the method requires the construction of an appropriate function, called customarily "the $\gg$-function".\smallskip

In case of the problem \ref{master} the $\gg$-function is a priori expressible in the form 
\begin{equation}\label{ansatz}
	\gg(z) = \frac 12 \vartheta(z;x,N) + \int_{z_0} ^z \big(P_{4N+2}(w)\big)^{\frac{1}{2}}\frac{\d w}{w^{2N+2}}+\frac{\ell}{2},\ \ \ \ell=\ell(x;N)\in\mathbb{C},\ \ \ z_0=z_0(x;N)\in\mathbb{C}
\end{equation}
where, in general, $P_{4N+2}(z)$ is an appropriate polynomial of the indicated degree. The ansatz \eqref{ansatz} is explained in the paragraph "Construction of the $\gg$-function'' of \cite{BB} and the discussion there can be applied almost verbatim here. From \eqref{ansatz} we see that the $\gg$-function is an Abelian integral on the Riemann surface of the square root of $P_{4N+2}(z)$; depending on the number of odd roots, this surface has a genus that can range from a minimum of $0$ (if there are only two simple roots in $P_{4N+2}(z)$) to a maximum of $2N$ (if all the roots are simple). Subsequently the Deift-Zhou analysis shows that 
\begin{quote}
If $x$ is such that the genus of the above Riemann surface is zero and the $\gg$-function satisfies the appropriate inequalities (recalled below), then the Riemann-Hilbert problem \ref{master} is eventually {\em solvable} for sufficiently large $n$.
\end{quote}
Therefore our strategy is as follows; 
we postulate a genus zero  Ansatz for the $\gg$-function in \eqref{y1}; the algebraic requirements are easily verified, but the required inequalities are not always verified. We shall then find where the inequalities fail, and hence where the roots are asymptotically confined.\smallskip
 
For the concrete construction of the $\gg$-function in the genus zero region we follow the logic outlined in \cite{BB}. We seek a function $y=y(z),z\in\mathbb{C}\backslash\mathcal{B}$  of the form 
\begin{equation}
	y(z)=\frac{1}{z^{2N+2}}\big(z^2+a^2\big)^{\frac{1}{2}}P(z;a),\ \ \ \ P\in\mathbb{C}[z]\ \ \ \ \textnormal{deg}(P)=2N,
	\label{y1}
\end{equation}
where $(z^2+a^2)^{\frac{1}{2}}$ is defined and analytic off the oriented branch cut $\mathcal{B}=\mathcal{B}(x,N)$ which connects the points $z=\pm \im a$. The precise location of $\mathcal{B}\subset\mathbb{C}$ shall be discussed in Section \ref{bcut} below, for now we require that $y$ satisfies the two conditions
\begin{equation}\label{g:1}
	y(z)=\frac{1}{2}\vartheta_z(z;x,N)+\mathcal{O}(1),\ \ z\rightarrow 0;\ \ \ \ \ y(z)=\frac{1}{z}+\mathcal{O}\left(z^{-2}\right),\ \ z\rightarrow\infty.
\end{equation}
Using simple algebra, we directly obtain
\bp The conditions \eqref{g:1} imply that $a$ and $x$ are related via
\begin{equation}\label{g:2}
	x=2a+\frac{c_N}{a^{2N}},\ \ \ \ c_N=\frac{(-1)^N}{2^{2N}}\binom{2N}{N}
\end{equation}
and the polynomial $P(z;a)$ is uniquely determined as
\begin{equation}\label{g:3}
	P(z;a)=z^{2N}-\frac{1}{2a}T_{N-1,-\frac{1}{2}}\left(\frac{z^2}{a^2}\right),
\end{equation}
where $T_{m,\alpha}(\z)$ is the Maclaurin polynomial of degree $m\in\mathbb{Z}_{\geq 0}$ of the function $(1+\z)^{\alpha} = 1+\mathcal{O}(\z),\z\rightarrow 0$.
\ep
\begin{proof} Observe that the condition $y(z)\sim\frac{1}{z},z\rightarrow\infty$ implies that $P(z;a)$ is monic and from the behavior at $z=0$ we find that
\begin{equation}\label{arrange}
	\frac{1}{2}\frac{1-xz^{2N}}{\sqrt{z^2+a^2}}=P(z;a)+\mathcal{O}\left(z^{2N+2}\right).
\end{equation}
Writing $P(z;a)=z^{2N}+Q(z)$ with a polynomial $Q(z)$ of degree at most $2N-1$ and reading \eqref{arrange} at $\mathcal{O}\left(z^{2N}\right)$, we get
\begin{equation*}
	\frac{1}{2\sqrt{z^2+a^2}}=Q(z)+\mathcal{O}\left(z^{2N}\right),\ \ z\rightarrow 0,
\end{equation*}
and thus
\begin{equation*}
	Q(z)=-\frac{1}{2a}T_{N-1,-\frac{1}{2}}\left(\frac{z^2}{a^2}\right)
\end{equation*}
which gives \eqref{g:3}. In order to deduce \eqref{g:2}, we recall that $P(z;a)$ is monic, i.e. we must have
\begin{equation*}
	x-\frac{1}{a^{2N}}\binom{-\frac{1}{2}}{N}=2a,
\end{equation*}
that is \eqref{g:2}.
\end{proof}
\bp
\label{Propstarcorners}
The branch points of the map $a=a(x)$, defined implicitly by  \eqref{g:2}, coincide with the values of $a(x)$ for which two zeros of $P(z)$ overlap with the branch points $z = \pm\im a$.
\ep
\begin{proof}
We have to evaluate the condition $P(\pm\im a)=0$; using \eqref{g:3} this amounts to 
\begin{align*}
	0=P(\pm\im a;a) = (-1)^N a^{2N} - \frac 1 {2a} \sum_{k=0}^{N-1}\binom{-\frac{1}{2}}{k}(-1)^k&=(-1)^Na^{2N}-(-1)^N\frac{N}{a}\binom{-\frac{1}{2}}{N}\\
&=(-1)^Na^{2N}-\frac{N}{a2^{2N}}\binom{2N}{N}.
\end{align*}
Thus the condition determining the coincidence of a zero of $P(z)$ with $z=\pm\im a$ is 
\be
\label{ss1}
 a^{2N+1}=\frac{(-1)^N N}{  2^{2N}} \le({2N\atop N}\ri) .
\ee
On the other hand the map \eqref{g:2} has a branch point where $x'(a)=0$, which gives exactly \eqref{ss1}.
\end{proof}
\subsection{The complex effective potential and the inequalities.}

For further steps it will prove useful to define the effective potential,
\begin{equation}
\label{ytophi}
	\varphi(z;a)=2\int_{\im a(x)}^zy(w)\d w,\ \ \ \ z\in\mathbb{C}\backslash\mathcal{B}
\end{equation}
which in the given situation \eqref{g:3} can be evaluated explicitly,
\begin{equation}\label{g:4}
	\varphi(z;a)=-2\ln\left(\frac{z+(z^2+a^2)^{\frac{1}{2}}}{\im a}\right)+\frac{2}{z}\big(z^2+a^2\big)^{\frac{1}{2}}-\frac{1}{2N+1}T_{N-1,-\frac{3}{2}}\left(\frac{z^2}{a^2}\right)\frac{(z^2+a^2)^{\frac{3}{2}}}{a^3z^{2N+1}}
\end{equation}
and all branches in \eqref{g:4} are principal ones such that $(z^2+a^2)^{\frac{1}{2}}\sim z$ as $z\rightarrow\infty$.
\bl
\label{Lemmader}
Given \eqref{g:4}, we have for $N\in\mathbb{Z}_{\geq 1}$,
\be
	\frac{\partial\varphi}{\partial x}(z;a) = -  \frac{(z^2 + a^2)^{\frac{1}{2}}}{za},\ \ z\in\mathbb{C}\backslash\mathcal{B}.
\ee
\el
\begin{proof} The jump of $\varphi(z;a)$ equals $4\pi\im$ on a contour that extends to infinity. Hence $\pa_x \varphi(z;a)$ has no jump on a contour which extends to infinity. Along $\mathcal{B}$ we have $(\pa_x \varphi(z;a))_+ = -(\pa_x\varphi(z;a))_-$. 
 Since $\varphi(z;a)$ vanishes at the branch point $z=\im a$ (and is constant $\pm 2\pi\im$ at $z=-\im a$ on the two sides) we deduce that $\pa_x \varphi(z;a)$ must be zero at $\pm\im a$.  Also (compare Section \ref{bcut} below),
 \be
 \frac{\partial\varphi}{\partial x}(z;a) = \begin{cases}
 -\frac{1}{z}+\mathcal{O}(1),&z\rightarrow 0\\
 \mathcal{O}(1),&z\rightarrow\infty.
 \end{cases}
 \ee
 Thus the ratio of the proposed expression for $\pa_x \varphi(z;a)$ is bounded at $z=\pm\im a$, analytic across the cut and bounded at $z=\infty$ with limit $1$ at $z=0$. The Proposition now follows from Liouville's theorem.
\end{proof}
The potential \eqref{ytophi} is related to the $\gg$-function \eqref{ansatz} by 
\be
\label{phitog}
 	\gg(z)= \frac 1 2 \big(\vartheta(z;x,N)-\varphi(z;a)+\ell\big),\ \ z\in\mathbb{C}\backslash\mathcal{B}
\ee
where the constant $\ell$ ({\em modified Robin constant}) is defined by the requirement that $\gg(z) = \ln (z) + \mathcal O(z^{-1})$ as $|z|\to \infty$. The relevant inequalities for $\gg(z)$ are more conveniently expressed directly as inequalities for the effective potential. 
In terms of the latter, 
 the following properties of the effective potential are equivalent to the existence of the $\gg$-function and characterize  $\varphi(z;a)$ (the proof of these statements is simple if not already obvious)
 \begin{enumerate}
 \item Near $z=0$ the effective potential has the behavior 
 \begin{equation}\label{phiatzero}
 \varphi(z;a)= -\vartheta(z;x) + \mathcal O(1) \ \ \Rightarrow \ \ 
 y(z) = \frac 1 2 \vartheta_z(z;x) + \mathcal O(1)\ ,
 \end{equation}
 while near $z=\infty$ it behaves as 
 \be
 \varphi (z) = -2\ln z + \mathcal O(1).\label{phiatinfty}
 \ee
 \item  Analytic continuation of $\varphi(z;a)$ in the domain $z\in\C\backslash\mathcal B$ yields the same function up to addition of  {\em imaginary} constants; in particular, the analytic continuation of $\varphi(z;a)$ around a large circle yields $\varphi(z;a) + 2\pi\im$;
 \item For each component $\mathcal B_j$ of $\mathcal B$  we have that,
 \begin{equation*}
 \varphi_+ (z;a) + \varphi_-(z;a) = -2\im\alpha_j ,\ \ \  \ z\in \mathcal B_j,\ \ \alpha_j\in\mathbb{R}
 \end{equation*}
\item The effective potential 
\begin{equation}\label{effective}
	\Phi(z;x)\equiv \Re\big(\varphi (z;a)\big),\ \ \ z\in\mathbb{C}\backslash\mathcal{B}
\end{equation}
with $a=a(x)$ as in \eqref{i:8} and \eqref{pick} is a harmonic function in $z\in\C\backslash\mathcal B$. Moreover $\Phi(z;\cdot) \big|_{\mathcal B} \equiv 0$. 
\item {\bf Inequality 1.} The sign of $\Phi(z)$ on the left and right of $\mathcal B$ is {\em negative}.
\item {\bf Inequality 2.} We can continuously deform the contour of integration $\gamma$ to a simple Jordan curve (still denoted by $\gamma$)  such that $\mathcal B \subset \gamma$ and such that $\Phi(z)\big|_{\gamma \setminus \mathcal B} >0$.
 \end{enumerate}
 Note that $\varphi(z;a)$ and $\gg(z)$ are both related to the  antiderivative of the differential 
 \begin{equation*}
 	y(z)\d z=\big(P_{4N+2}(z)\big)^{\frac{1}{2}}\frac{\d z}{z^{2N+2}}
\end{equation*}
which is defined on a Riemann surface $X=\{(w,z):\,w^2=P_{4N+2}(z)\}$ of genus between $0$ and $2N$. Since $\Phi(z) = \Re(\varphi(z))$ vanishes along $\mathcal B$, it also follows that $\mathcal B$ is a subset of its zero level set; therefore, $\mathcal B$ consists of an union of arcs defined locally by the differential equation $\Re(y(z) \d z) =0$.

\subsection{Location of branch cut}\label{bcut}
The following Proposition appeared in \cite{BB} but applies also to the present situation \eqref{g:4}.
\bp 
\label{cutBg0}
The effective potential $\Phi(z;a)=\Re\big(\varphi(z;a)\big)$ has the following properties:
\begin{enumerate}
	\item The function $\Phi(z;a)$ is defined modulo a sign depending on the determination of $(z^2+a^2)^{\frac{1}{2}}$.
	\item The zero-level set $\mathcal{Z}=\{z\in\mathbb{C}: \Phi(z;a)=0\}$ is well defined independent of the determination of the square root in (1) and invariant under the reflection $z\mapsto -z$.
	\item For $|a|$ sufficiently large there are two smooth branches of the zero-level set $\mathcal{Z}$ which connect $z=\pm \im  a$ and which are symmetric under $z\mapsto -z$.
\end{enumerate}
\ep
\begin{proof} Statements $(1)$ and $(2)$ follow just as in (\cite{BB} Proposition 3.5), for $(3)$ we note that, as $a\rightarrow\infty$
\begin{equation}
	\varphi(za;a)\rightarrow -2\ln\left(\frac{z+(z^2+1)^{\frac{1}{2}}}{\im}\right)+\frac{2}{z}(z^2+1)^{\frac{1}{2}}=-2\int_{\im}^z\big(1+w^2\big)^{\frac{1}{2}}\,\frac{\d w}{w^2}\equiv Q(z)
\end{equation}
and the limit is uniform on compact subsets of the Riemann sphere not containing $z=0$. The remaining logic is now as in \cite{BB}.
\end{proof}
Suppose that $a>0$ is sufficiently large and thus Proposition \ref{cutBg0} applies. We claim that $\mathcal{B}=\mathcal{B}(x,N)$ is the branch in point $(3)$ above which intersects the positive half ray $\mathbb{R}_+$ (by deformation this fixes the branch cut for all $x\notin P_N$). In order to see this, recall from \eqref{phiatzero}, as $z\rightarrow 0$,
\begin{equation*}
	y(z)\sim-\frac{1}{2}\vartheta_z(z;x,N)\sim\frac{1}{2z^{2N+2}}\ \ \ \ \Rightarrow\ \ \ y(z)\rightarrow+\infty,\ \ z\downarrow 0.
\end{equation*}
But this requires in \eqref{y1} that $(z^2+a^2)^{\frac{1}{2}}\sim -a$ as $z\rightarrow0$. Simultaneously \eqref{phiatinfty} requires $y(z)\sim\frac{1}{z}$ near $z=\infty$ and hence $(z^2+a^2)^{\frac{1}{2}}\sim z$ as $z\rightarrow+\infty$. Hence, the determination of the square root in \eqref{y1} has to change on the positive half ray, i.e. $\mathcal{B}=\mathcal{B}(x,N)$ is as claimed.

\subsection{The inequalities of the $\gg$-function and the region $P_N$}\label{gfunctdelta} Since the quadratic differentials
\begin{equation}
\label{etag}
	\eta=(\d\varphi)^2=4(z^2+a^2)P^2(z;a)\frac{\d z^2}{z^{4N+4}},\ \ \ \ \textnormal{and in general}\ \ \ \ \ \eta=(\d\varphi)^2=4P_{4N+2}(z)\frac{\d z^2}{z^{4N+4}},
\end{equation}
are of the type studied by Jenkins and Spencer \cite{JS}, that is, of the form $R(z) \d z^2$ with $R(z)$ a rational function, we can follow some of the reasoning which was  already explained in \cite{BB}.\smallskip

{\bf Preliminaries.}  Define the set  $\mathfrak H_x$ to consist of the union of the second order poles and all  {\it critical  trajectories}, i.e.,   all solutions of $\Re(\d \varphi(z;a))= \Re(2y(z) \d z)=0$ that issue from each of the zeros and simple poles of $R(z)$; the latter are absent in our case. The zeros are at $\pm\im a$ and at the $N$ pairs $\pm z_{j}$ which are the roots of the even polynomial  $P(z;a)$. Also \cite{S}, there are 
 {\em $2k+1$} branches of $\mathfrak H_x$  issuing from each of the points of order $k$ of $R(z)$, $k=-1,0,1,...$ (the case $k=-1$ corresponds to simple poles, and all others to zeroes). We are interested in the connected components of
 \begin{equation*}
 	\C\backslash\mathfrak H_x=\bigsqcup_jK_j
\end{equation*}
and a simple argument in analytic function theory (see  \cite{JS}) shows that each {\em simply} connected component $K_j$ is conformally mapped by $\varphi(z;a)$ into a half-plane or a vertical strip $\alpha<\Phi(z;x)<\beta$; each {\em doubly} connected component $K$ is mapped to an annulus (or a punctured disk) $\{
r_-<|w|<r_+\}$ by $w = {\rm e}^{\frac{2\pi\im} {p} \varphi(z)}$ where $p= \oint_\gamma \d \varphi$  and $\gamma$ is a closed simple contour separating the two boundary components of $K$. It is also shown in \cite{JS} that there are no other possibilities for the topology of the connected components $K_j$. Moreover, there is a one-to-one correspondence between annular domains (including the degenerate case of a punctured disk) and  free homotopy classes of simple closed contours $\gamma$ for which $\oint_\gamma \d \varphi \neq 0$. In our case there is only one such class corresponding to a loop encircling the origin, and hence only one annular domain which we denote by $K_\infty$ (which is actually a punctured disk).\smallskip

By construction, $\varphi(z;a)\sim -2\ln z, z\rightarrow\infty$ which shows that $z=\infty$ is at the center of a conformal punctured disk via the conformal map $w = {\rm e}^{\frac{1}{2}\varphi(z)}$. Moreover the level sets $C_r= \{z:\,\Phi(z;x) = -2\ln r\}$ are foliating a region around $z=\infty$ in topological circles if $r$ is sufficiently large. Thus none of the hyperelliptic trajectories issuing from $\pm\im a, \{\pm z_j\}$ can ``escape'' to infinity; they either connect to $z=0$ or amongst each other.
Suppose $r_0$ is the infimum of the $r>0$ for which $C_r$ is smooth; this means that $C_{r_0}$ contains at least one zero of $\d_z\varphi$ (by symmetry, it contains then two zeros in our situation). The annular (punctured disk) domain $K_\infty$ is then (see Figure \ref{Kinfty})
\be
K_\infty= \ov {\bigcup_{r>r_0} C_{r}}.
\ee
We denote also $D_0= \ov {\C \setminus K_\infty}$, which is a simply connected, symmetric region containing the origin.\smallskip

\noindent\begin{minipage}{0.59\textwidth}
{\bf Necessary and sufficient condition for the correct inequalities in genus zero.} 
We argue that we need to have $r_0=0$. To put it differently, the ``first encounter'' of the level sets $C_r$ as $r $ decreases must be  with the two branch points $\pm\im a$ rather than any of the zeros $\{z_j\}$. 
We shall then verify that this occurs for $x>0$ large enough.\\

 {\bf Sufficiency}. Suppose now that $r_0=0$ and thus $\pm\im a \in K_\infty$ and $\pm z_j \in\textnormal{Int}(D_0)$.
 Then the simple, closed loop $\pa K_\infty$ is separated into two components by $\pm\im a$ and each of them is an hyperelliptic trajectory. 
We know that there must be three trajectories from each $\pm\im a$ and two of them are already accounted for and  form the boundary of $D_0$ (see Fig. \ref{Kinfty});
 thus the third trajectory is entirely contained in $D_0$, which is   compact. 
 
 Now let us turn our attention to $D_0$; the points $\pm z_j \in D_0$ for $j=1,\dots, N$. In $D_0$ each branch of $y(z)$ \eqref{y1} is single valued (the branch points are on the boundary of $D_0$). 
 Only one of the two branches of $y(z)$ has the behavior $\frac 1 2 \vartheta_z(z;x)$; integrating this branch from $\im a$ coincides with $\varphi(z;x)$ in $D_0$. The value of the sign of $\Phi$ in the interior of $D_0$ close to the boundary $D_0$ determines which of the two parts of $\pa D_0\backslash\{\pm\im a\}$ is the branch cut $\mathcal B$: this is the part which has $\Phi>0$ on {\em both} sides (ie. in $D_0$ and $K_\infty$). Thus $\Phi$ is continuous but not harmonic on $\mathcal B$, while on $\pa D_0 \backslash\mathcal B$ it is continuous and harmonic.
We still need to show that there is a path connecting $\pm\im a$ and which lies within the region $\Phi <0$. \end{minipage}\hspace{0.02\textwidth}\begin{minipage}[h]{0.39\textwidth}
\includegraphics[width=1\textwidth]{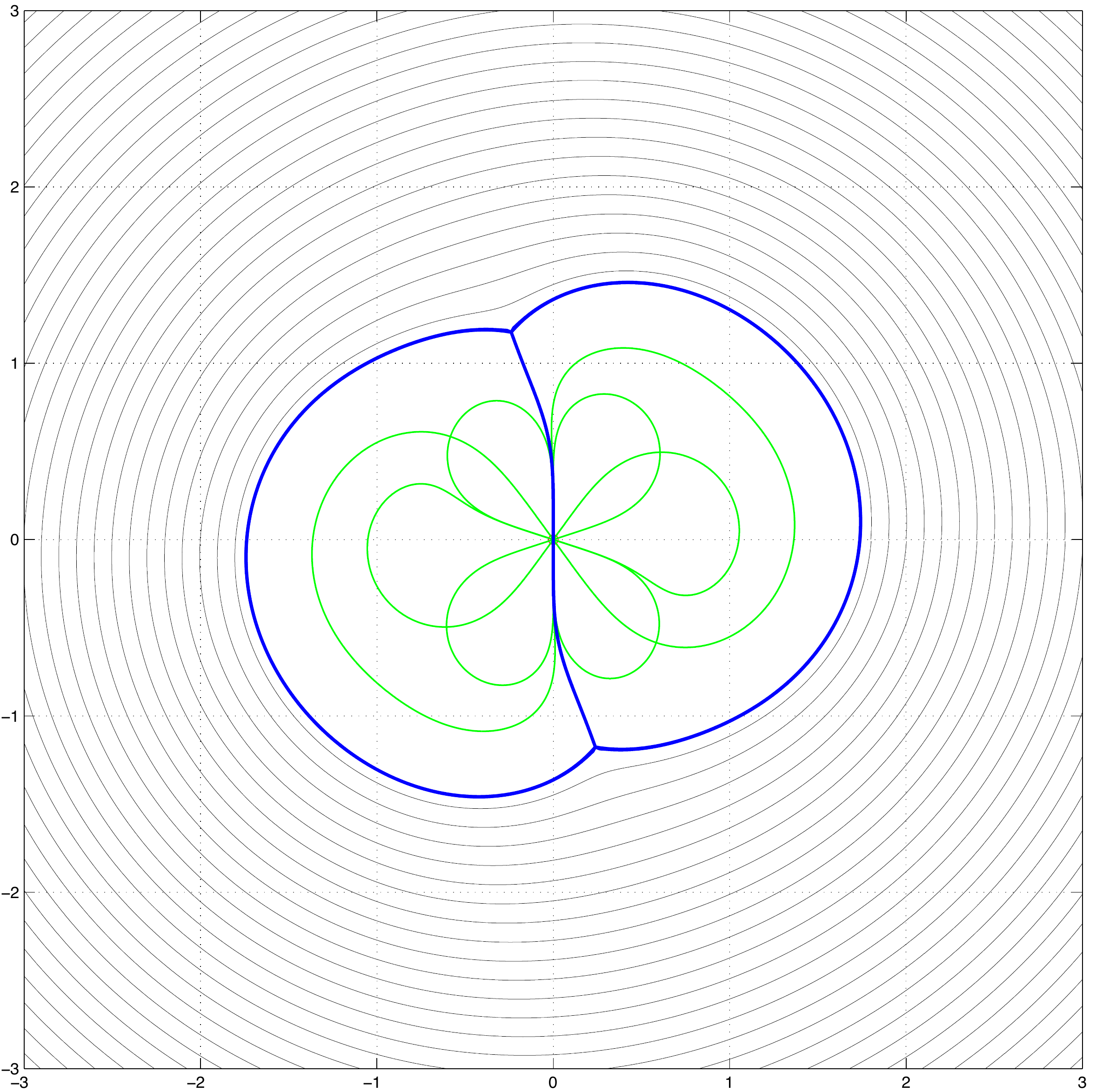}
\captionof{figure}{}{
 Illustration of the conformal punctured disk $K_\infty$, foliated by the trajectories $C_r$. The complement, $D_0$, contains the other critical trajectories. In this example $N=2$ so there are $4$ saddle points inside $D_0$, visible here where the critical trajectories intersect at right angles.}
  \label{Kinfty}
\end{minipage}
\smallskip

This follows from the topological description of the possible regions $K_j$ discussed in the paragraph "Preliminaries". Indeed let $K_1$ be the region containing the arc $\pa D_0\backslash\mathcal B$ where $\varphi(z;a)$ is conformally one-to-one. From the discussion of signs thus far, this is either a half-plane $w = \Phi  <0$ or a strip $-\epsilon< \Phi  <0$ (the only annular domain is $K_\infty$). The two points $\pm\im a$ are mapped on the imaginary axis $\Re w = \Phi =0$; thus there is a path connecting $\varphi(\im a)$ to $\varphi(-\im a)$ in the $w$-plane which lies in the left half plane. The pre-image of this path in the $z$ plane connects thus $\pm\im a$ and $\Phi $ restricted to the interior points of this path is strictly negative.
\smallskip

 {\bf Necessity.} If $r_0>0$ then $\pm\im a\in D_0$. The trajectories issuing from $\pm\im a$ all belong to the zero level set of $\Phi$. None of them can connect to any of the zeros $\{\pm z_j\}$, and thus they either connect to each other or to the origin. 
 Since the sign of $\Phi $ changes $2N+1$ times around $z=0$, they all must go to the origin and thus there is no possibility of deforming the contour of integration so that it contains the branch cut $\mathcal B$ and avoids the origin.\smallskip
 
{\bf Sufficient condition for the correct inequalities in highest genus.}
 We work with the same general setup as in the previous case.
 Now the quadratic differential is of the form on the right in \eqref{etag}. Suppose that $P_{4N+2}(z)$ there has all simple roots $\{a_{k}^\pm\}_{k=1}^{2N+1}$ (the roots come clearly in pairs of opposite signs). We claim that a sufficient condition for the fulfillment of the inequalities  is that {\bf all} branch points $a_k^{\pm}$ lie on $\pa D_0 = \pa K_\infty$.
In this case $\pa K_\infty$ is broken into $4N+2$ arcs (see for example Figure \ref{maxgenus}). There is only one branch of $y(z)$ that behaves as $y(z)\sim\frac 12 \vartheta_z(z;x)$ near $z=0$; the integral of this branch with base point $a_1^+$ is single--valued in $D_0 = \C \backslash K_\infty$ because the region contains no branch points and the residue of $y(z)$ at $z=0$ vanishes; this integral then defines $\varphi(z;a)$ (and $\Phi$) within $D_0$. The level curves of $\Phi(z;x)$ that issue from $a_j$ and do not connect to other branch points must connect to the origin because $\Phi(z;x)$ changes sign exactly $4N+2$ times when going around the origin. The regions where $\varphi(z;a)$ is now one-to-one within $D_0$ are $4N+2$ half-planes because their boundary has only one connected component. Necessarily in $2N+1$ of them $\Phi(z;x)<0$ and $2N+1$ of them $\Phi(z;x) >0$. The arcs of $\pa K_{\infty} \backslash \{a_k^{\pm}\}_{k=1}^{2N+1}$ bounding the three regions where $\Phi(z;x)<0$ are the cuts and the other are simply zero level sets separating regions where $\Phi(z;x)$ has opposite signs. The possibility of connecting two branch points that are connected by an arc of these level sets follows exactly by the same argument used in the previous paragraph.\smallskip

{\bf Occurrence of the necessary/sufficient conditions.} 
By point {\bf (3)} of Proposition  \ref{cutBg0}, for $a$ (hence, $x$) large enough there is a smooth branch of the zero level set of $\Phi(z;x)$ that connects $\pm\im a$; by symmetry, there is another one and thus the third branch of the level set that issue from $\pm\im a$ must go to zero (we have seen that there is no branch that extends to infinity). The remaining roots of $P(z;a)$ all tend to zero as $a\to \infty$ (which is easily seen from the explicit expression \eqref{g:3}). Thus they must fall within the region $D_0$. Then the necessary condition in genus $0$  is fulfilled. For the case of maximal genus, the occurrence of the sufficient condition is contained in Proposition \ref{propgmax}.\\

{\bf The discriminant locus}.
By discriminant locus we refer to the boundary of the locus in the $x$--plane where the inequalities for the genus zero Ansatz fail. This is the boundary of a region $P_N$; from the discussion above it follows 
 that the inequalities are preserved under a deformation in $x$, until failure occurs exactly when  one of the zeros of $P(z;a)$ falls on the branch cut $\mathcal{B}$ connecting $\pm\im a$, and hence, by symmetry, one also intersecting the opposite arc. We know that this does not happen for $|x|$ sufficiently large and hence the discriminant must be a bounded set.\smallskip
 
 In order to detect the occurrence of the situation above it is {\em necessary} (but not sufficient) that 
 $\Phi(z_j;x) =0$ for some $j$, i.e., one of the  saddle points of $\Phi$ lies on the zero-level set;
 \be
 \pa P_N \subset  
 \mathfrak Z_N=
  \le\{ 
 x\in \C:\ \exists\,z\in \C: 
\Re\big(\varphi(z;a(x))\big) = 0,\ 
\varphi_{zz}(z;a(x)) = 0
 \ri\}
 \label{condzero}.
 \ee 
 The set $\mathfrak Z_N$ is clearly closed, and thus $\pa P_N$ must be compact (since we know already it is bounded). However, the set $\mathfrak Z_N$ is strictly larger than $\pa P_N$; indeed it describes the situation where any of the saddle points of $\Phi(z;x)$ intersect any branch of the zero level set; The zero level set contains several branches besides the branch cut and hence the set $\mathfrak Z_N$ in \eqref{condzero} describes also all these ``fake'' situations.\smallskip
 
A detailed analysis for arbitrary $N$ seems unwieldy. We shall attempt below only a partial study of the case $N=2,3$ in Appendix \ref{outer23}, where we show that the points \eqref{starcorners} do indeed belong to the boundary of the regions $P_N$ for $N=2,3$. However the set $\mathfrak Z_N$ is easily drawn and the results are displayed in Figure \ref{Stars} together with the roots of some higher polynomials. The result of this discussion is the following theorem;
 \begin{theo}
 The roots $\mathfrak R_n^{[N]}$  of the polynomials $Q_n^{[N]} \le( n^{\frac {2N}{2N+1}} x\ri)$ lie all within an arbitrarily small neighborhood of a compact region $P_N$ as $n\to \infty$;  the boundary of this region consists  entirely  of  a finite union of real-analytic arcs in the $x$ plane satisfying the condition \eqref{condzero}.
 \end{theo}
 The condition \eqref{condzero} is spelt out in more detail in the statement of Theorem \ref{thm:boundary}, which is henceforth proved as well.
 \br
 A careful consideration should allow also to prove that the region is simply connected. 
 It is also relatively simple to show that $x=0$ belongs to the interior of this region (see Proposition \ref{propgmax} below).
 \er
 
 The set $\mathfrak Z_N$ in \eqref{condzero} contains the points $x$ for which one pair of roots of $P(z;a)$ coincides with the branch points $\pm\im a$; these points are easily computed and are precisely the $2N+1$ points in \eqref{starcorners}. However we cannot positively conclude for general $N$ that they are on the boundary of $P_N$, although this is quite evident from the numerics. Also, the detailed shape of $\pa P_N$, beyond the easily established discrete $\Z_{2N+1}$ symmetry, is hard to describe in more detail; for example it is not obvious how to conclude that it consists of $4N+2$ smooth arcs for $N\in\mathbb{Z}_{\geq 2}$, as the Figure \ref{Stars} clearly shows.
 We find it however already sufficiently interesting that we can narrow down the boundary of $\Delta$ as a subset  of a simple set of equations \eqref{condzero}, although we cannot completely describe it.

\subsection{At the center of $P_N$} In a small vicinity of $x=0$, we have
\begin{equation*}
	P_{4N+2}(z)=\prod_{k=1}^{2N+1}(z-a_k^+)(z-a_k^-)
\end{equation*}
with $a_j^{\pm}\neq a_k^{\pm}$ for $j\neq k$. As in \cite{BB}, the branch points $a_k^+=-a_k^-$ are partially determined through \eqref{phiatzero},\eqref{phiatinfty} and in addition through Boutroux  type conditions
\begin{equation*}
	\Re\left(\oint_{\gamma_j}y(z)\d z\right)=0,\ \ \ \gamma_j\in H_1(X,\mathbb{Z});\ \ y(z)=\frac{1}{z^{2N+2}}\big(P_{4N+2}(z)\big)^{\frac{1}{2}}.
\end{equation*}
The latter are imposed on the hyper elliptic curve $X=\{(w,z):\,w^2=P_{4N+2}(z)\}$ which is obtained by crosswise gluing together two copies of $\mathbb{C}\backslash\mathcal{B}$ with $\mathcal{B}=\bigcup_{k=1}^{N}[a_{2k-1}^+,a_{2k}^+]\cup[a_{2N+1}^+,a_1^-]\cup\bigcup_{k=1}^N[a_{2k}^-,a_{2k+1}^-]$. Solvability of the resulting system for $\{a_k^+\}_{k=1}^{2N+1}$ would now follow as in \cite{BB}, but here we are only interested in the case $x=0$.
\bp[Compare Proposition 3.9 in \cite{BB}]
\label{propgmax}
For $x=0$ the $\gg$--function is obtained from \eqref{phitog}, \eqref{ytophi} using  
\begin{equation*}
y(z) = \frac 1{z^{2N+2}} \big(R(z)\big)^{\frac{1}{2}}\ ,\ \ \qquad R(z) = P_{4N+2}(z)\Big|_{x=0}= z^{4N+2} + \frac 1 4
\end{equation*}
which is defined and analytic off $z\in\mathcal{B}$ with the branch points $a_{k}^+=a_{k,0} = 2^{-\frac 1{2N+1}}{\rm e}^{\frac {\im\pi k}{2N+1}},k=1,\ldots,2N+1$.
\ep
\begin{proof}
{\bf Local behavior.} Near $z=0$ we have 
\be
y(z) = -\frac 1 2\frac {1} {z^{2N+2}}  \le(1 + \mathcal O\left(z^{4N+2}\right)\ri)
\ee
and near infinity clearly $y(z) = \frac 1 z + \mathcal O(z^{-2})$. Note that the determination of the root near $z=0$ is the opposite.\smallskip

{\bf Boutroux condition.} We have 
\be
\int_{a_{j,0}}^{a_{j+1,0}} \big(R(z)\big)^{\frac{1}{2}}\frac{\d z}{z^{2N+2}} = \int_{\omega a_{j,0}}^{\omega a_{j+1,0}} \omega^{-2N-1} \big(R(\omega z)\big)^{\frac{1}{2}}\frac { \d z}{z^{2N+2}} =- \int_{a_{j+1,0}}^{a_{j+2,0}} \big(R(z)\big)^{\frac{1}{2}}\frac { \d z}{z^{2N+2}} ,
\ee
and thus it is sufficient to verify the Boutroux condition 
\begin{equation*}
	\oint_{\gamma} \big(R(z)\big)^{\frac{1}{2}} \frac {\d z}{z^{2N+2}}\in\im\R
\end{equation*}	
for a specific $j\in\{1,\ldots,2N\}$. This condition guarantees that all branch points lie in the zero level set of $\Phi$.  But for $j=1$ it follows immediately that the integral is imaginary using the Schwartz symmetry.\smallskip

{\bf Connectedness of the level curves.}
First of all the set $\Phi(z)=0$ in $\C\backslash\{0\}$ consists of one connected component alone; this is so because there are no saddle points and if there were two or more connected components, there would have to be a saddle point in the region bounded by them.
We shall now verify that the level curves satisfy the necessary and sufficient conditions specified in Section \ref{gfunctdelta}. The critical trajectories  must
\begin{enumerate}
\item connect all $4N+2$ branch points
\item obey the $\Z_{4N+2}$ symmetry because of obvious symmetry.
\end{enumerate}
A simple counting then shows that the only possibility is that exactly one trajectory from each branch point (in fact a straight segment) connects the branch points to $0$ because the sign of $\Phi$ changes $4N+2$  times around a small circle surrounding the origin.  The other two trajectories must then connect the branch points. This is depicted in  Figure \ref{maxgenus}. The discussion on the necessary and sufficient condition for the correct inequalities is now as explained in (\cite{BB}, Section  3.1).
\end{proof}

\begin{figure}
\includegraphics[width=0.3\textwidth]{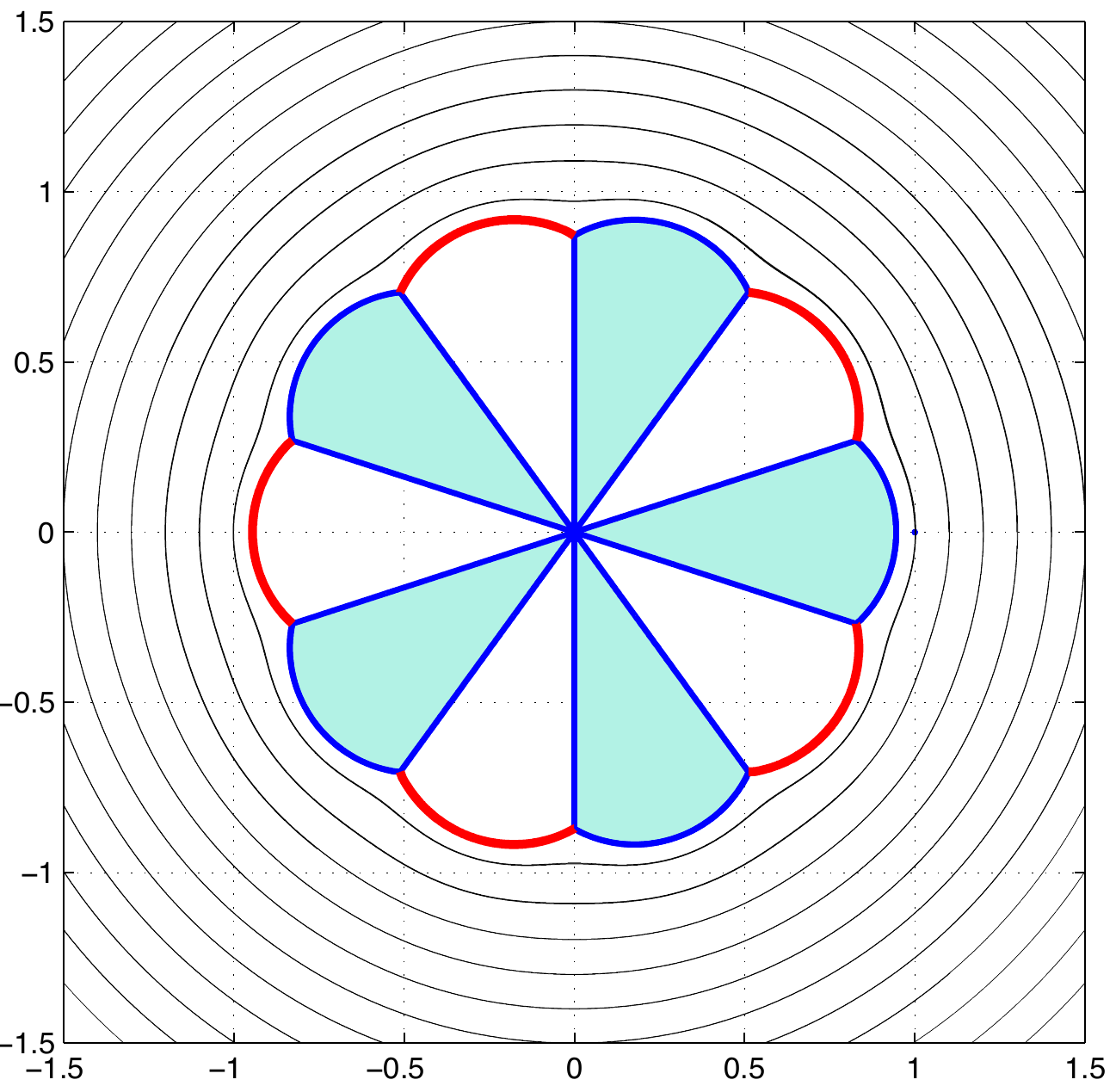}
\hspace{2cm}
\includegraphics[width=0.3\textwidth]{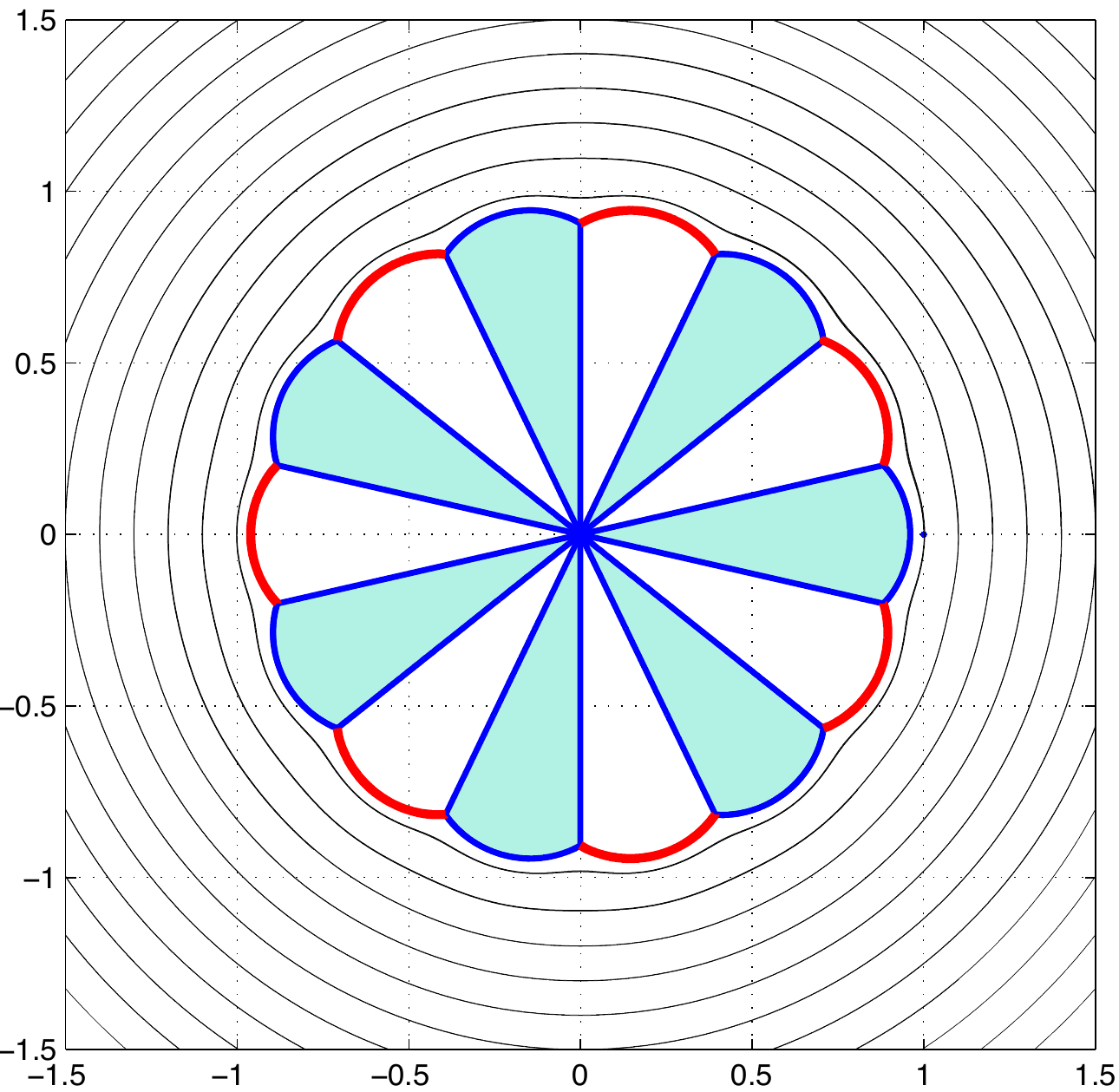}
\caption{The placement of the trajectories of the quadratic differential $\eta$ in the case of $x=0$; the red arcs are the branch cuts of the the $\gg$--function. The cases presented correspond to $N=2,3$ (left/right). Also indicated is the foliation by the trajectories $C_r$ of the region $K_\infty$ (see Section \ref{gfunctdelta}).
The shaded (cyan) regions indicate where $\Phi>0$; all the trajectories that issue from the branch points constitute the set $\Phi=0$.
 }
\label{maxgenus}
\end{figure}

\appendix
\section{Proof of the Miura relation \eqref{Miura2}}
\label{appmiura}
\br
We draw the reader's attention to the various notations used in this section,
\be
\bt = (t_1,t_2,t_3,t_4,\dots)\ ,\ \ \ 
\bt_o= (t_1,0,t_3,0,t_5,\dots)\ ,\ \ \ 
\wt{\bt}_o= (t_1, 0,2^2 t_3, 0, 2^4 t_5,\dots),\ \ \ t_j\in\mathbb{C}.
\ee
\er
Let $k,\ell\in\mathbb{Z}_{\geq 0}$ and introduce
\begin{equation*}
\mu_{k}(\bt) = \frac{1}{2\pi\im}\oint_{S} z^k {\rm e}^{\vartheta(z)}\frac {\d z}z,\ \ 
\vartheta(z) = \vartheta(z;\bt) = \sum_{j\geq 1} \frac {t_j}{z^j},\ \  \ \ \ 
\Delta_{n,\ell}(
\bt) = \det\big[\mu_{\ell + j+k-2}(\bt)\big]_{j,k=1}^n,\ \ \Delta_{0,\ell}(\bt)\equiv 1
\end{equation*}
where $S\subset\mathbb{C}$ denotes the unit circle traversed in counterclockwise orientation. Recalling \eqref{hdef} we see that
\begin{equation*}
	\mu_k(\bt)=\frac{1}{k!}\frac{\d^k}{\d w^k}\exp\Big[\sum_{j\geq 1}t_jw^j\Big]\bigg|_{w=0}=h_k(\bt)
\end{equation*}
and thus with \eqref{JTrudi},
\begin{equation*}
	\Delta_{n,\ell}(\bt)=s_{(\ell+n+1)^n}(\bt).
\end{equation*}
In particular, by Lemma \ref{lemmadouble}, we know that for the special value $\ell=0$ we have the identity 
\begin{equation*}
	\Delta_{n,0} (\bt_o) =s_{(n+1)^n}(\bt_o)=2^{-n^2}s_{\delta_n}^2\big(\,\wt{\bt}_o\big).
\end{equation*} 
Next, let $\{p_{n,\ell}(z)\}_{n\geq 0}$ be the monic orthogonal polynomials associated with the measure 
\begin{equation*}
	\d\nu_{\ell}(z)=\frac{1}{2\pi\im}z^\ell  {\rm e}^{\vartheta(z)}\frac {\d z}z,\ \ z\in S,\ \ell\in\mathbb{Z}_{\geq 0};\ \ \ \ \ \oint_Sp_{n,\ell}(z)p_{m,\ell}(z)\d\nu_{\ell}(z)=\widehat{h}_n\delta_{nm}.
\end{equation*} 
It is well known \cite{D} that the matrix 
\begin{equation}\label{GAMMA}
\Gamma_{n,\ell}(z) = \begin{bmatrix}
p_{n,\ell}(z) & \frac{1}{2\pi\im}\oint_Sp_{n,\ell}(w)\frac{\d\nu_{\ell}(w)}{w-z}\\
 \gamma_{n-1,\ell}\,p_{n-1,\ell}(z) & \frac{\gamma_{n-1,\ell}}{2\pi\im}\oint_Sp_{n-1,\ell}(w)\frac{\d\nu_{\ell}(w)}{w-z}\\
\end{bmatrix},\ \ z\in\mathbb{C}\backslash S;\ \ \ \ \gamma_{n,\ell}=-2\pi\im\frac{\Delta_{n,\ell}(\bt)}{\Delta_{n+1,\ell}(\bt)}
\end{equation}

satisfies a Riemann--Hilbert problem, i.e. $\Gamma_{n,\ell}(z)$ is analytic for $z\in\mathbb{C}\backslash S$ and we have the conditions
\begin{equation*}
\big(\Gamma_{n,\ell}(z)\big)_+ = \big(\Gamma_{n,\ell}(z)\big)_- \begin{bmatrix}
1 & z^{\ell-1}{\rm e}^{\vartheta(z)}\\
0 & 1
\end{bmatrix},\ \ z\in S;\ \ \ \ \Gamma_{n,\ell}(z)=\left(I+\Gamma'_{n,\ell}(\infty)\frac{1}{z}+\mathcal{O}\left(z^{-2}\right)\right)z^{n\sigma_3},\ z\rightarrow\infty.
\end{equation*}

\bp
\label{Satoid}
The following identities hold for the Hankel determinants $\Delta_{n,\ell}(\bt)$.
\begin{equation}
\frac{\Delta_{n,\ell+1}( \bt-[z])}{\Delta_{n, \ell}(\bt)} = (-1)^n\big(\Gamma_{n,\ell}(z)\big)_{11} \ \qquad \ \ 
\frac{\Delta_{n,\ell-1}( \bt+[z])}{\Delta_{n, \ell}(\bt)} = (-1)^n\big(\Gamma_{n,\ell}(z)\big) _{22}
\label{satodelta}
\end{equation}
and
\begin{equation}
\frac{\Delta_{n+1,\ell}(\bt)}{\Delta_{n,\ell}(\bt)}  = -2\pi\im\,\big(\Gamma_{n,\ell}'(\infty)\big)_{12},\ \qquad \ \
\frac{\Delta_{n-1,\ell}(\bt)}{\Delta_{n,\ell}(\bt)}  = \frac{\im}{2\pi}\,\big(\Gamma_{n,\ell}'(\infty)\big)_{21}
\label{standard}
\end{equation}
where $[z]$ denotes the infinite vector of components $(z, \frac {z^2}2, \frac {z^3}3, \frac {z^4}4,\ldots)$, i.e. 
\begin{equation*}
	\bt\mp[z]=\left(t_1\mp z,t_2\mp\frac{z^2}{2},t_3\mp\frac{z^3}{3},\ldots\right).
\end{equation*}
\ep
\begin{proof}
The two identities in \eqref{standard} follow simply by inspection of the expression \eqref{GAMMA}. 
As for the identities \eqref{satodelta}, 
the  proof follows from Heine's formula for the orthogonal polynomials and the observation that 
\begin{equation*}
	w^{\ell} \exp\big[\vartheta(w; \bt - [z])\big] =  w^{\ell-1} (w-z)  \exp\big[\vartheta(w; \bt)\big],\ \ \ \ 
	w^{\ell} \exp\big[\vartheta(w; \bt + [z])\big] = \frac{w^{\ell+1}}{w-z}  \exp\big[\vartheta(w; \bt)\big].
\end{equation*}
Indeed, we have
\begin{align*}
	\Delta_{n,\ell+1}(\bt-[z])=\det\big[\mu_{\ell+j+k-1}(\bt-[z])\big] _{j,k=1}^n= \frac 1{n!} \oint_{S^n} \prod_{ j<k} (w_j - w_k)^2 \prod_{j=1}^n w_j^{\ell} \exp\big[
	\vartheta(w_j; \bt -[z])\big]\frac{\d w_j}{2\pi\im}\\
	=\frac 1{n!}  \oint_{S^n} \prod_{j<k} (w_j - w_k)^2 \prod_{j=1}^n (w_j-z) \d\nu_{\ell}(w_j)= (-1)^n\det\begin{bmatrix}
	\mu_{\ell} &\cdots & \mu_{\ell+n}\\
	\vdots &  & \vdots\\
	\mu_{\ell+n-1} &  & \mu_{\ell+2n-1}\\
	1 &  \cdots & z^n
	\end{bmatrix}= (-1)^np_{n,\ell}(z) \Delta_{n, \ell}(\bt) 
\end{align*}
where we used the well-known representation of orthogonal polynomials in terms of moment determinants (see, e.g. Proposition 3.8 in \cite{D}).
The second identity can be found in \cite{BDS}, but we can give here a direct derivation using Andreief's identity \cite{A}. Recall that $\d\nu_{\ell}(w)=w^{\ell-1}{\rm e}^{\vartheta(w)}\d w$. Then 
\begin{align}
	\oint_{S^n} \prod_{j<k}(w_k-w_j)^2 \prod_{j=1}^{n} \frac {\d\nu_{\ell}(w_j)}{w_j-z} = \oint_{S^n} \det \Big[w_j^{k-1}\Big]_{j,k=1}^n 
	\det\Big[w_j^{k-1}\Big]_{j,k=1}^n \prod_{j=1}^{n} \frac {\d\nu_{\ell}(w_j)}{w_j-z}\nonumber\\
	=\oint_{S^n} \det \Big[w_j^{k-1} \Big]_{j,k=1}^n\det \bigg[\frac{w_j^{k-1}}{w_j-z}\bigg]_{j,k=1}^n\prod_{j=1}^{n} {\d\nu_{\ell}(w_j)} \label{132}
\end{align}
Multi-linearity allows us to replace the monic powers in the first determinant by the monic orthogonal polynomials $p_{j,\ell}(w)$, so that we obtain
\begin{equation}\label{1322}
	\eqref{132}=\oint_{S^n} \det\Big[p_{k-1,\ell}(w_j)\Big]_{j,k=1}^n\det \bigg[\frac{w_j^{k-1}}{w_j-z} \bigg]_{j,k=1}^n\prod_{j=1}^{n} {\d\nu_{\ell}(w_j)}.
\end{equation}
Now, in the second determinant we can subtract to the columns $2\leq k\leq n$ the multiple $z^{k-1}/(w_k-z)$ of the first column, thus obtaining 
\begin{equation}\label{13222}
\eqref{1322} =  
\oint_{S^n} \det\Big[p_{j-1,\ell}(w_k)\Big]_{j,k=1}^n 
 \det\begin{bmatrix}
 \frac{1}{w_1-z} & \frac{w_1-z}{w_1-z} & \cdots & \frac{w_1^{n-1}-z^{n-1}}{w_1-z}\\
 \vdots & \vdots& &\vdots\\
 \frac{1}{w_n-z} & \frac{w_n-z}{w_n-z} & \cdots & \frac{w_n^{n-1}-z^{n-1}}{w_n-z}
 \end{bmatrix} \prod_{j=1}^n\d\nu_{\ell}(w_j)
\end{equation}
Using now Andreief's identity we obtain  
\begin{equation*}
\eqref{13222} = n! \det \begin{bmatrix}
\oint_Sp_{0,\ell}(w)\frac{\d\nu_{\ell}(w)}{w-z} & \oint_Sp_{0,\ell}(w)\frac{w-z}{w-z}\d\nu_{\ell}(w) & \cdots & \oint_Sp_{0,\ell}(w)\frac{w^{n-1}-z^{n-1}}{w-z}\d\nu_{\ell}(w)\\
\oint_Sp_{1,\ell}(w)\frac{\d\nu_{\ell}(w)}{w-z} & \oint_Sp_{1,\ell}(w)\frac{w-z}{w-z}\d\nu_{\ell}(w) & & \oint_Sp_{1,\ell}(w)\frac{w^{n-1}-z^{n-1}}{w-z}\d\nu_{\ell}(w)\\
\vdots & \vdots& & \vdots\\
\oint_Sp_{n-1,\ell}(w)\frac{\d\nu_{\ell}(w)}{w-z} &\oint_Sp_{n-1,\ell}(w)\frac{w-z}{w-z}\d\nu_{\ell}(w) & \cdots& \oint_Sp_{n-1,\ell}(w)\frac{w^{n-1}-z^{n-1}}{w-z}\d\nu_{\ell}(w)
\end{bmatrix},
\end{equation*}
but due to orthogonality the matrix above has the following structure
\begin{equation*}
\eqref{13222}= n! \det \begin{bmatrix}
\oint_Sp_{0,\ell}(w)\frac{\d\nu_{\ell}(w)}{w-z} & \widehat{h}_0 & \star & \cdots & \star\\
\oint_Sp_{1,\ell}(w)\frac{\d\nu_{\ell}(w)}{w-z} & 0 & \widehat{h}_1 & \star & \star\\
\vdots & \vdots & 0& \ddots&\star\\
\vdots & \vdots &\vdots & \ddots&\widehat{h}_{n-2}\\
\oint_Sp_{n-1,\ell}(w)\frac{\d\nu_{\ell}(w)}{w-z} & 0 & 0 & \cdots & 0
\end{bmatrix}=
(-1)^{n+1}n!  \oint_S p_{n-1,\ell}(w)\frac{\d\nu_{\ell}(w) }{w-z}\prod_{j=0}^{n-2}\widehat{h}_j.
\end{equation*}
However
\begin{equation*}
	\widehat{h}_j = \oint_{S} p_{j,\ell}^2(w) \d \nu_{\ell}(w) = \frac {\Delta_{j+1,\ell}(\bt)}{\Delta_{j, \ell}(\bt)}
\end{equation*}
and therefore
\begin{equation*}
	\Delta_{n,\ell-2}(\bt +[z]) =(-1)^{n+1} \frac{\Delta_{n-1,\ell}(\bt)}{\Delta_{0,\ell}(\bt)}\oint_S p_{n-1,\ell}(w)\frac{\d\nu_{\ell}(w) }{w-z}=(-1)^n\Delta_{n,\ell}(\bt)\big(\Gamma_{n,\ell}(z)\big)_{22}.
\end{equation*}

\end{proof}
\subsection{Dodgson-Hirota bilinear identity}
Consider the following matrix--valued function
\begin{equation*}
H_{n,\ell}(z;\bt, {\bf s}) =  \Gamma_{n,\ell-1}(z;\bt)\begin{bmatrix}
{\rm e}^{\vartheta(z;\bt)-\vartheta(z;{\bf s})} & 0\\
0 & z^2
\end{bmatrix}\Gamma_{n,\ell+1}^{-1}(z;{\bf s}),\ \ z\in\mathbb{C}\backslash(S\cup\{0\}).
\end{equation*}
A direct inspection using the jumps of $\Gamma_{n,\ell}$ shows that this matrix has no jumps on the contour $S$, however an essential singularity at $z=0$ due to the presence of the exponentials. We can thus compute the contour integral below in two ways. First by evaluation as a residue at infinity; 
\begin{equation}\label{res1comp}
\frac{1}{2\pi\im}\oint_{|z|=R}H_{n,\ell}(z;\bt, {\bf s})\frac {\d z}{z} = \begin{bmatrix}
 1-\big(\Gamma_{n,\ell-1}'(\infty;\bt )\big)_{12} \big(\Gamma_{n,\ell+1}'(\infty;{\bf s})\big)_{21} & \star \\
\star & \star
\end{bmatrix}
\end{equation}
where the $\star$ indicates expressions which are not relevant to the steps below. Secondly we evaluate the left hand side in \eqref{res1comp} as a residue at $z=0$, but we are only interested in the $(11)$-entry,
\begin{equation*}
	\frac{1}{2\pi\im}\oint_{|z|=R}\big(H_{n,\ell}(z;\bt, {\bf s})\big)_{11}\frac {\d z}{z}=\res_{z=0}\frac{1}{z}{\rm e}^{ \vartheta(z;\un t) - \vartheta(z;\un s) }\big(\Gamma_{n,\ell-1}(z;\bt)\big)_{11}\big(\Gamma_{n,\ell+1}(z;{\bf s})\big)_{22}.
\end{equation*}
Hence with \eqref{res1comp} and Proposition \eqref{Satoid},
\begin{equation*}
	\res_{z=0}\frac{1}{z}{\rm e}^{ \vartheta(z;\un t) - \vartheta(z;\un s) }\frac{\Delta_{n,\ell}(\bt-[z])\Delta_{n,\ell}({\bf s}+[z])}{\Delta_{n,\ell-1}(\bt)\Delta_{n,\ell+1}({\bf s})}=
	1-\frac{\Delta_{n+1,\ell-1}(\bt)\Delta_{n-1,\ell+1}({\bf s})}{\Delta_{n,\ell-1}(\bt)\Delta_{n,\ell+1}({\bf s})},
\end{equation*}
or equivalently
\begin{equation}\label{dog}
	\res_{z=0}\frac{1}{z}{\rm e}^{ \vartheta(z;\un t) - \vartheta(z;\un s) }\Delta_{n,\ell}(\bt-[z])\Delta_{n,\ell}({\bf s}+[z])=\Delta_{n,\ell-1}(\bt)\Delta_{n,\ell+1}({\bf s})-\Delta_{n+1,\ell-1}(\bt)\Delta_{n-1,\ell+1}({\bf s}).
\end{equation}
\br Identity \eqref{dog} closely resembles a ``Hirota" version of the classical Dodgson determinantal identity, for if we set $\bt ={\bf s}$ then \eqref{dog} reduces to the Dodgson identity for Hankel determinants,
\begin{equation}\label{Dodgson}
	\Delta_{n,\ell}^2 = \Delta_{n,\ell-1}\, \Delta_{n,\ell+1}- {\Delta_{n+1,\ell-1} \Delta_{n-1,\ell+1}}.
\end{equation}
\er
We now rewrite equation \eqref{dog} with the substitution $\bt \mapsto \bt+ {\bf h},{\bf s} = \bt-{\bf h}$ and define
\begin{align*} \label{MasterHirotaDodgson}
	H\!D_{n,\ell} (\bt,{\bf h}) &= \res_{z=0}\le(\frac {1}{z} {\rm e}^{2 \vartheta(z;{\bf h})} \Delta_{n,\ell}(\bt+{\bf h} -[z]) \Delta_{n, \ell}(\bt-{\bf h}  + [z])\ri)  \\
	& - \Delta_{n,\ell-1}(\bt+{\bf h}) \Delta_{n,\ell+1}(\bt-{\bf h})- \Delta_{n+1,\ell-1}(\bt+{\bf h}) \Delta_{n-1,\ell+1}(\bt-{\bf h})
\end{align*}
so that \eqref{dog} can be written in the compact form
\be
\label{HDnl}
H\!D_{n,\ell} (\bt,{\bf h}) \equiv 0 \ ,\ \ \ \forall\, \bt, {\bf h},\ \ \ \ \ \forall\, n,\ell\in\mathbb{Z}_{\geq 1}.
\ee
For the rest of this section we shall set all even times to zero, i.e. we choose $\bt=\bt_o$. Now use Corollary \ref{corfreak} in conjunction with \eqref{Dodgson},
\begin{equation}\label{freak2}
\Delta_{n,1}^2(\bt_o) = \Delta_{n,0}(\bt_o) \Delta_{n,2}(\bt_o) - \Delta_{n+1,0}(\bt_o) \Delta_{n-1,2}(\bt_o)=2(-1)^n\Delta_{n,0}(\bt_o)\Delta_{n+1,0}(\bt_o),
\end{equation}
and recall Lemma \ref{lemmadouble},
\begin{equation*}
	\Delta_{n,0}(\bt_o)=s_{(n+1)^n}(\bt_o)=2^{-n^2}s_{\delta_n}^2(\wt{\bt}_o).
\end{equation*}
Hence with \eqref{gW} for $t_{2j+1}=0,j>N$ and $t_1=x$,
\begin{equation}\label{idcheck}
	\Delta_{n,0}(\bt_o)=2^{-n^2}{\rm e}^{2g_n(x;\un t)},\ \ \ \ \ \ \ \Delta_{n+1,0}(\bt_o)=2^{-(n+1)^2}{\rm e}^{2W_n(x;\un t)+2g_n(x;\un t)}
\end{equation}
Differentiating \eqref{HDnl} with respect to $h_j$ we can derive a whole hierarchy of equations, however we are only interested in one particular identity:
\bea
\nonumber
\frac{\partial^2}{\partial h_1^2}H\!D_{n,\ell} (\bt_o,{\bf h})\Big|_{{\bf h}={\bf 0}} = -     {\frac {\partial ^2 \Delta_{n+1,\ell-1}  }{\partial t_{{1}}^2}}          \Delta_{{n-1, \ell +1}}      
-   \Delta_{{n+1, \ell -1}}        {\frac {\partial ^2\Delta_{n-1,\ell+1} }{\partial t_{{1}}^2}}     
+2\,   {\frac {\partial   \Delta_{{n+1, \ell -1}} }{\partial t_{{1}}}}          {\frac {\partial   \Delta_{{n-1, \ell +1}} }{\partial t_{{1}}}}       
\\+     {\frac {\partial ^2  \Delta_{n,\ell-1}  }{\partial t_{{1}}^2}}          \Delta_{{n, \ell +1}}      
-2\,   {\frac {\partial   \Delta_{{n, \ell -1}}   }{\partial t_{{1}}}}        {\frac {\partial  \Delta_{{n, \ell +1}} }{\partial t_{{1}}}}        
+   \Delta_{{n, \ell -1}}       {\frac {\partial ^2  \Delta_{n,\ell+1}    }{\partial t_{{1}}^2}}  
+2\,   {\frac {\partial ^{2} \ln \Delta_{n,\ell} }{\partial {t_{{1}}}^{2}}}           ( \Delta_{{n, \ell }})^2 =0  
\label{premiura}
\eea
and the argument of all determinants in the right hand side equals $\bt=\bt_o$. For $\ell=1$, with \eqref{premiura} and \eqref{freak}, \eqref{freak2} this leads to
\begin{equation*}
0= \le(\Delta_{n+1,0} '' \Delta _{n,0}  +\Delta _{n+1,0} \Delta_{n,0}'' - 2 \Delta_{n+1, 0}
' \Delta_{n,0}' \ri)  +\le(\ln \Delta_{n+1,0} +   \ln \Delta_{n,0} \ri)'' \Delta_{n+1,0} \Delta_{n,0},\ \ (')=\frac{\partial}{\partial t_1}
\end{equation*}
which can be rewritten as 
\begin{eqnarray*}
	0&=& \le(\frac{\Delta_{n+1,0} ''}{\Delta_{n+1,0}}   +\frac{\Delta_{n,0}''}{\Delta_{n,0}} - 2\frac{ \Delta_{n+1, 0}'}{\Delta_{n+1,0}}
	\frac{\Delta_{n,0}' }{\Delta_{n,0}}\ri)  +\le(\ln \Delta_{n+1,0} +   \ln \Delta_{n,0} \ri)''\\
	&=&2\frac{\partial^2}{\partial t_1^2}\ln\big(\Delta_{n,0}\Delta_{n+1,0}\big) +\le(\frac{\partial}{\partial t_1}\ln \frac{\Delta_{n,0}}{\Delta_{n+1,0}}\ri)^2,
\end{eqnarray*}
and after simplification with \eqref{idcheck},
\begin{equation*}
	2\,\partial_x^2g_n(x;\un t)=-\partial_x^2W_n(x;\un t)-\big(\partial_xW_n(x;\un t)\big)^2
\end{equation*}
which completes the proof of \eqref{Miura2}.
\section{The outer corners of the regions $P_N$ for $N=2,3$}
\label{outer23}
In this section we offer a proof that the points \eqref{starcorners} belong to the boundary of $P_N$. 
The proof is a verification that the inequalities for the effective potential are fulfilled at the particular values of $a(x)$ determined in \eqref{ss1}. These correspond in the $a$-plane to the points \eqref{starcorners} in the $x$-plane. The proof is a simple deformation argument starting from large $|a|$ (and hence also large $x$).\smallskip

Observing various panes in Figure \ref{Stars} and using the $\Z_{2N+1}$ symmetry of the region, it is sufficient to show that the point 
\be
a_{0}^{[N]}= \frac{1}2 (-1)^N \le(2N\le({2N\atop N}\ri)\ri)^{\frac 1{2N+1}} \ \ \Rightarrow \ \ \ 
x_{0}^{[N]} = (-1)^N\le((2N+1)\left(\frac{2N+1}{2N}\right)^{2N}\binom{2N}{N}\ri)^{\frac 1{2N+1}}
\ee
\begin{wrapfigure}{r}{0.43\textwidth}
\includegraphics[width=0.4\textwidth]{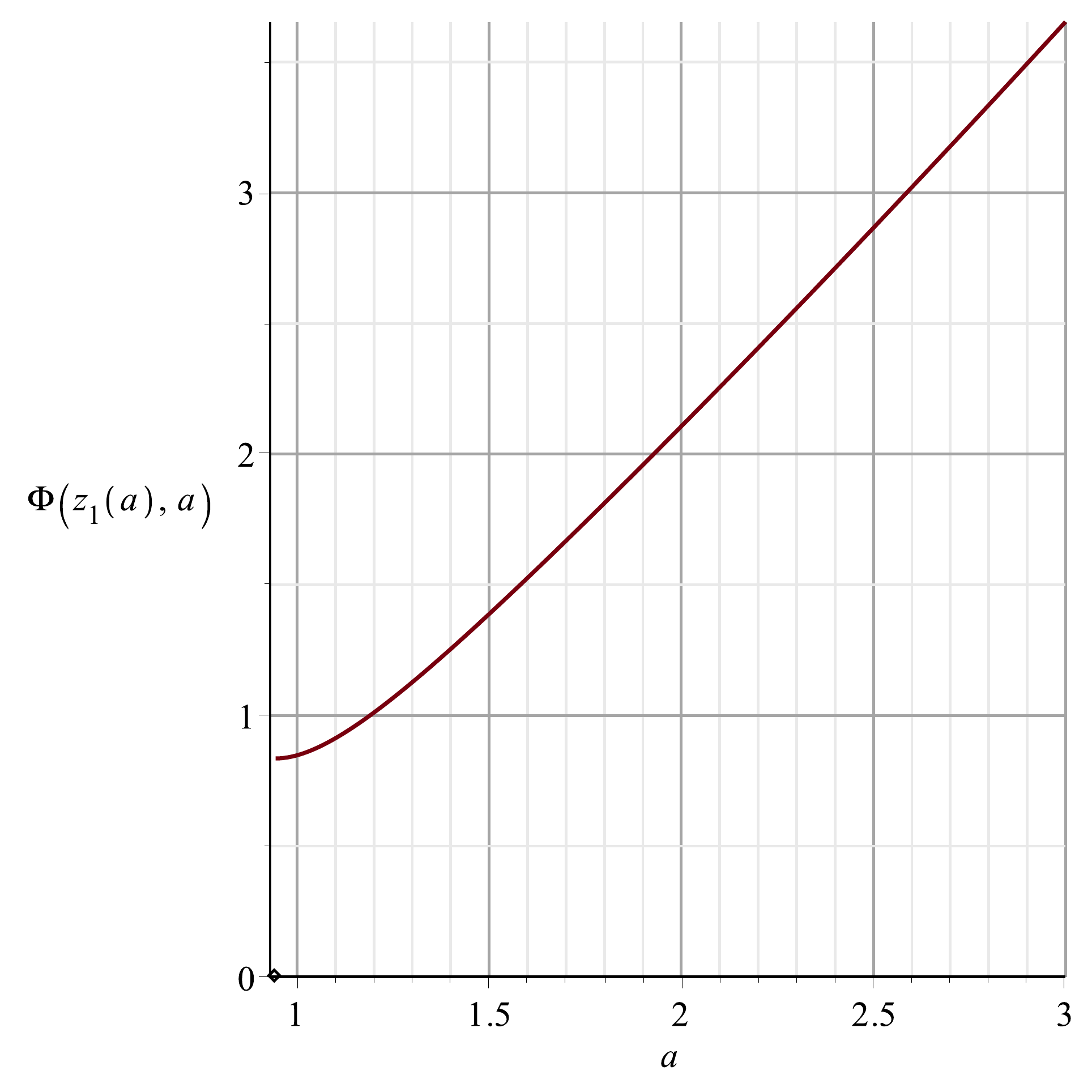}
\caption{The graph of the value of $\Phi$ at the saddle point $z_1$.}
\label{RootDescent}
\end{wrapfigure}
(or rather its $x$--image) belongs to the boundary of $P_N$. 
This point is alternatively positive or negative, depending on the parity of $N$. Consider now in some detail the case $N=2$; then $a_{0}^{[2]} \simeq 0.944$ ($x_{0}^{[2]} \simeq 2.36021$).
In this case the polynomial $P(z;a)$ \eqref{g:3} equals

\begin{minipage}{0.57\textwidth}
\begin{equation*}
	P=P_2(z;a) = z^4 + \frac 1 4 \frac {z^2}{a^3} - \frac 1 {2a}.
\end{equation*}
\end{minipage}\\[2pt]

Let $z_j^\pm (a), \ j=1,\dots 2$ denote the roots of  $P_2$. We know from the argument in Section  \ref{gfunctdelta} that for $|a|$ large the inequalities are fulfilled; as we deform $a$ from larger absolute values to smaller ones,  these inequalities can fail only if the sign of $\Phi(z_j^\pm (a); a)$ changes.\smallskip 

We now simply have to verify that the sign of $\Phi (z_j(a);a)$ remains constant as $a$ decreases from $+\infty$ to the critical value $a_{0}^{[2]}$ (corresponding to $x$ decreasing from $+\infty$ to the rightmost corner $x_{0}^{[2]}$). Since the four roots admit an explicit expression in terms of $a$, this verification is a simple exercise  in calculus. To be more precise, one pair that we denote $z_2^\pm(a)$ is purely imaginary and   lies on the zero level set of $\Phi(z;a)$ identically for $a\in [a_{0}^{[2]}, \infty )$; this is not a cause of concern because it belongs to the level curve (in fact a straight line) joining $z=\pm\im a$ to $z=0$. 
The other pair $z_1^\pm(a)$ is real for $a\in [a_{0}^{[2]}, \infty )$. Then one can easily verify that 
\begin{equation*}
F(a)= \Phi\big(z_1^\pm (a);a\big)
\end{equation*}
is, depending on which of the two member of the pair,  a monotone increasing/decreasing function of $a\in [a_{0}^{[2]}, \infty ) $ and not changing sign. This verification uses Lemma \ref{Lemmader} and the explicit expression for the roots, so that (for $a$ real) 
\be
\label{signphi}
\frac {\d}{\d a} \Phi\big(z_j(a);a\big)  =  - \Re \left({\pa_x}\Phi(z;a) \frac {\d x}{\d a} \frac {\pa_a P(z;a)}{P'(z;a)} \bigg|_{z=z_j(a)}\right)
= 
\Re\left(\frac {\sqrt{z^2 + a^2}}{za}  \frac {\pa_a P(z;a)}{P'(z;a)} \bigg|_{z=z_j(a)}\frac {\d x}{\d a}\right)
\ee
In Figure \ref{RootDescent} we display the graph of $\Phi(z_1(a);a)$ in the range $[a_{0}^{[2]},\infty)$; the monotonicity can be shown by inspecting the sign of \eqref{signphi}; we leave the detail to the reader.
The argument above can be repeated for $N=3$, but for larger $N$ we were not able to find a unifying argument. 

\end{document}